\documentclass[american]{article}
\usepackage[T1]{fontenc}
\usepackage[latin9]{inputenc}
\usepackage{geometry}
\usepackage{array}
\usepackage{mathtools}
\usepackage{amsmath}
\usepackage{amsthm}
\usepackage{amssymb}
\usepackage{comment}
\usepackage{graphicx}
\usepackage{setspace}
\usepackage[authoryear]{natbib}
\makeatletter

\providecommand{\tabularnewline}{\\}

\theoremstyle{plain}
\newtheorem{thm}{\protect\theoremname}
\theoremstyle{plain}
\newtheorem{prop}[thm]{\protect\propositionname}
\theoremstyle{plain}
\newtheorem{lem}[thm]{\protect\lemmaname}

\usepackage{tikz}
\usepackage{libertine}
\usepackage{xcolor}
\usepackage[final]{changes}

\setcitestyle{round,comma}
\usepackage{rotating}
\usepackage{subfig}

\definecolor{green}{RGB}{0,150,0}

\@ifundefined{showcaptionsetup}{}
{\PassOptionsToPackage{caption=false}{subfig}}
\usepackage{subfig}
\makeatother
\usepackage{babel}
\providecommand{\lemmaname}{Lemma}
\providecommand{\propositionname}{Proposition}
\providecommand{\theoremname}{Theorem}

\begin{document}
\global\long\def\N{\mathbb{N}}

\global\long\def\R{\mathbb{R}}

\global\long\def\P#1{P\left(#1\right)}

\global\long\def\E#1{\mathbb{E}\left[#1\right]}

\global\long\def\hyp#1#2{\,_{#1}F_{#2}}

\global\long\def\var#1{\text{var}\left[#1\right]}

\global\long\def\argmax{\operatornamewithlimits{argmax}}

\global\long\def\simoperator{\operatornamewithlimits{\sim}}

\global\long\def\ind#1{\boldsymbol{I}\left(#1\right)}

\global\long\def\simlim#1{\simoperator\limits _{#1}}

\global\long\def\pochhammer#1#2{#1^{\overline{#2}}}

\title{Exact steady-state distributions of multispecies
  birth-death-immigration processes: effects of mutations and
  carrying capacity on diversity}

\author{Renaud Dessalles$^{1,2}$, Maria D'Orsogna$^{1,3}$, Tom Chou$^{1,4}$ \\
$^1$Department of Biomathematics, University of California, Los Angeles, CA 90095-1766 \\
$^2$MaIAGE, INRA, Universit\'{e} Paris-Saclay, 78350 Jouy-en-Josas, France \\
$^3$Department of Mathematics, California State University, Northridge, CA 91330 \\
$^4$Department of Mathematics,  University of California, Los Angeles, CA 90095-1555}

\maketitle
\begin{abstract}
Stochastic models that incorporate birth, death and immigration (also
called birth-death and innovation models) are ubiquitous and
applicable to many problems such as quantifying species sizes in
ecological populations, describing gene family sizes, modeling
lymphocyte evolution in the body. Many of these applications involve
the immigration of new species into the system. We consider the full
high-dimensional stochastic process associated with multispecies
birth-death-immigration and present a number of exact and asymptotic
results at steady-state.  We further include random mutations or
interactions through a carrying capacity and find the statistics of
the total number of individuals, the total number of species, the
species size distribution, and various diversity indices. Our results
include a rigorous analysis of the behavior of these systems in the
fast immigration limit which shows that of the different diversity
indices, the species richness is best able to distinguish different
types of birth-death-immigration models. We also find that detailed
balance is preserved in the simple noninteracting
birth-death-immigration model and the birth-death-immigration model
with carrying capacity implemented through death. Surprisingly, when
carrying capacity is implemented through the birth rate, detailed
balance is violated.
\end{abstract}

\section{Introduction}

In recent years, stochastic Birth-Death-Immigration (BDI) models have
emerged as effective descriptions of the evolution of multi-species
populations. BDI models assume that each individual belongs to a given
``species'' and undergoes a classical birth-death process; offspring
populate the same species as their parent, while new species are
introduced via immigration\deleted{, or, in some
  variants,} \added{and/or} mutation. The body of work on BDI models in
the mathematical, ecological and biological literature is rich, and
\added{many} results have been \added{independently}
discovered \deleted{independently} in the context of different
disciplines.  Arguably, the first BDI model can be traced to
\citet{Karlin1967} who described the evolution of different alleles in
a genetic population.  Later, similar tools were applied to ecology in
the context of the ``Neutral Theory of Biodiversity''
\citep{Hubbell2001,MacArthur2016,Lambert2011,Volkov2003}, where BDI
models were used to study the abundance distribution of island
populations that undergo continuous immigration from the mainland.
BDI models have also been used under the name birth, death and
innovation models by \citet{Karev2002} to describe gene domain family
size in genomes. Here, each domain is part of a family, and can be
duplicated or deleted; new domains of new families can be added to the
genome via horizontal gene transfer. \citet{Desponds2016} and
\citet{Lythe2016} have instead employed BDI formalisms to study
lymphocyte populations in an organism. T-cells expressing the same
surface receptor are assumed to belong to the same ``clone'' (the
species).  Each T-cell can divide, generating
\added{receptor-}identical daughter cells, and die through
apoptosis. In this context, immigration is represented by the export
of new na\"{i}ve T-cells from the thymus. Due to the
\added{large number of} theoretically possible T-cell
clonotypes that can be generated, with estimates ranging from
$10^{15}-10^{20}$ \citep{PRICE}, one can assume that each new export
almost surely generates a new clone rather than contribute to an
existing clonotype. Another application of BDI models
\added{arises in} the study of microbiota populations
in the gut of metazoa \citep{Sala2016}. Finally, counting ``clones''
\added{is also used} in stochastic models of nucleation, where a high-
or infinite-dimensional distribution function can be used to describe
states comprised of certain numbers of clusters (the ``clones'') of
specific size \citep{HETERO2011,JCP2012,HETERO2013}. In the rest of
this paper, we will use both ``individuals'' (or ``particles'') and
``species'' to describe the two types of quantities (individuals of a
given species \added{and} the number of species of a
given size) in all of the above-mentioned examples.

\deleted{The goal of this paper is to introduce a unifying
  mathematical framework to describe general BDI models; we will study
  steady states and related measures of diversity and
  entropy. Particularizations and constraints can be added at will,
  but we identify three major modeling representations that can serve
  as building blocks for more complex systems.}

Note that we will focus exclusively on ``neutral'' BDI representations
in the sense of the Neutral Theory of Biodiversity
\citep{Bell2001,Hubbell2001}, that is all individuals within a
population are subject to the same birth and death rates so that there
is no fitness difference in the population. Our first model is the
simple BDI \added{(sBDI)} model where each individual evolves
independently of all others and where the only possible processes are
birth, death and immigration. The second model \added{(BDIM)} further
includes mutations, whereby the dynamics of each individual is still
uncoupled from that of others, but where new species can arise via
mutations. The last model \added{(BDICC)} includes a carrying capacity
  on the death rate to represent the sharing of limited resources.  In
  this case, the dynamics of each individual is coupled to that of
  others, and the overall mathematical analysis is more complex.
  Thus, for simplicity, when including a carrying capacity term, we
  exclude mutations. The three major BDI processes we will analyze are
  depicted in Figure~\ref{Fig1}.

\added{Since measures of diversity in a population are also of
  significant interest in ecology
  \citep{PALMER1990,COLWELL1994,CHAO2014}, we will also analyze
  species diversity through three commonly used indexes
  (\citet{Morris2014}): the species richness (the total number of
  species in the system), Shannon's entropy, and Simpson's diversity
  index, and we will contrast and compare these quantities among the
  different models.}

\deleted{Some aspects or versions of these models have already been
  studied, such as in Karlin and McGregor (1967), Travar\'{e} (1989),
  Lambert (2011), Karev et al. (2002) for simple
  birth-death-immigration, or in Lambert (2011) for the
  birth-death- immigration process with mutation. However, a
  systematic analysis of steady states, a thorough comparison of these
  models in terms of species diversity, and an analysis of their
  behavior in limiting cases are still lacking. In this paper, we
  provide an accessible, yet rigorous, theoretical analysis for each
  of the three types of BDI models outlined above. In particular, we
  will determine the conditions for the existence for an equilibrium
  distribution, and we will provide analytical expressions for the
  distribution of the total number of individuals, the total number of
  species, as well as the species size distributions.}

\added{The goal of this paper is to provide an accessible, yet
  rigorous, theoretical analysis of each of the three types of BDI
  models outlined above.  In particular, we determine the conditions
  for the existence of an equilibrium distribution and derive
  analytical expressions for the steady state distributions of the
  total number of individuals, the numbers if clones of each size, the
  total number of species in the system, as well as the expected
  species diversities predicted by model. Some results presented here
  can be recovered from previous work.  More precisely, time-dependent
  versions of our sBDI model can be found in
  \citet{Karlin1967,Travare1989,Lambert2011} and one particular
  version the BDIM model (with somatic mutations) is described in
  \citep{Lambert2011}. In each case, it is possible to recover the
  steady state distributions of the total population and the total
  number of species by evaluating the infinite-time limits of their
  results.}

\begin{figure}
\begin{center}
\includegraphics[width=6.4in]{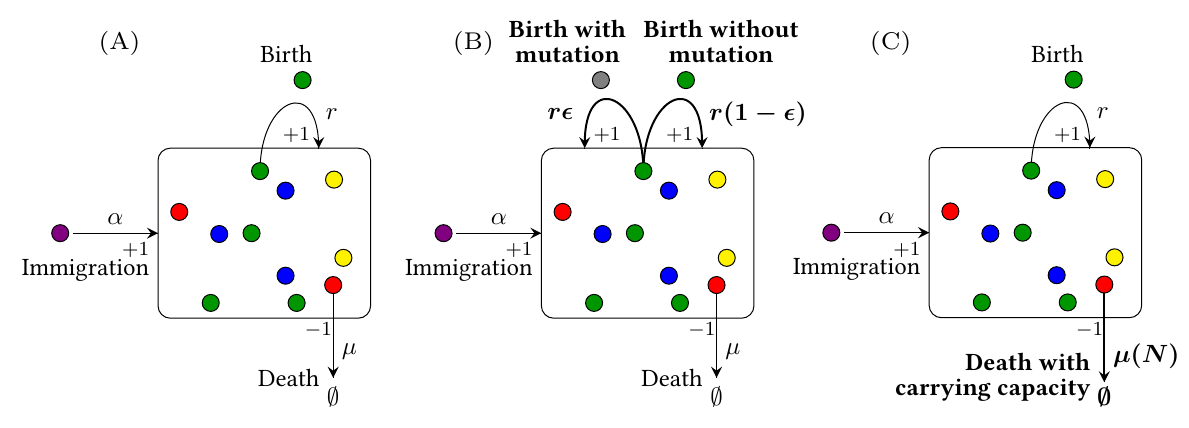}
\caption{\label{Fig1}Schematic of various Birth-Death-Immigration
  processes. Three distinct variants are considered, including (a)
  Simple birth-death-immigration (sBDI) without mutation, where
  $r,\mu$ are constants, (b) Birth-death-immigration with mutation
  (BDIM) where $r,\mu$ are constants and mutation rate $\epsilon > 0$,
  and (c) Birth-death-immigration without mutation but with carrying
  capacity (BDICC) where $r$ is constant and $\mu = \mu(N)$ is an
  increasing function of total population $N$. We will also analyze a
  variant of the BDICC model, the BDICC-bis model, where $\mu$ is
  constant but the growth $r=r(N)$ is assumed to be a decreasing
  function of the total population $N$.}
\end{center}
\end{figure}

\added{Our work provides a number of additional results in the steady
  state limit: i) the theoretical analysis of an interacting BDI model
  with carrying capacity is new, to the best of our knowledge; ii) we
  provide ``full probability distributions'' that completely describe
  the properties of each model and that can be used to derive general
  quantities of interest, (in particular, the moments of the species
  counts); iii) we analytically quantify diversity indices predicted
  by each model; iv) we provide systematic quantitative comparisons
  between the models and v) we derive simpler limiting forms of our
  results in the large immigration rate limit. A summary of all our
  results can be found in Table~\ref{tab:recap} in the general case
  and in Table~\ref{tab:recap-1} for the fast immigration limit.  The
  interested reader will find more details of the methods and the
  proofs of the derivation of these expressions in the Mathematical Appendix.}

\section{\label{sec:basic-assumptions}Basic definitions}

In this section we outline some general assumptions and introduce the
mathematical notation to describe our three BDI models. First, we
assume new individuals immigrate to the system following a Poisson
point process of rate $\alpha$, i.e. the time interval between
successive immigration events is given by a random variable
exponentially distributed with rate $\alpha$. Each arriving individual
will define a new species not yet present in the system. The random
variable representing the total number of individuals in the system
will be denoted by $N$ and the total number of species by $C$. We
consider both ``particle-count'' and ``species-count'' representations
($n_{i}$ and $c_{k}$ respectively) of the system: in the particle
count representation, $n_{i}$ (with $1\leq i\leq C$) represents the
number of individuals in the $i$-th species; in the species-count
representation, $c_{k}$ (with $k\geq1$) represents the number of
species having exactly $k$ individuals. In principle, there can be
species with infinite population and hence both $n_{i}$ and the index
$k$ are unbounded. The sequence of all numbers $(n_{i})_{i\leq C}$,
the infinite vector $\vec{c}=(c_{k})_{k\geq1}$, as well as $N$ and
$C$, are related by the following expressions:

\begin{equation}
c_{k}=\sum_{i=1}^{C}\ind{n_{i},k}\quad\text{for }k\geq1, 
\qquad C=\sum_{k\geq1}c_{k},\qquad N=\sum_{i=1}^{C}n_{i}=\sum_{k\geq1}k\,c_{k},
\label{eq:relations NCc}
\end{equation}
where $\boldsymbol{I}$ is the indicator function such that
$\ind{x,y}=1$ if $x=y$ and $0$ otherwise. Effectively, the first
relation in Equation~\ref{eq:relations NCc} will count the number of
species that carry $k$ individuals. The second relation describes the
total number of clones $C$ that are present in the system, while $N$
is the total number of particles in the system. For many applications
$C$ and $N$ are large. For example, in humans, the richness of naive
T-cells $C\sim 10^{6}-10^{10}$ \citep{ASQUITH,Lythe2016}.

The particle-count and species-count representations are related since
a given sequence $(n_{i})_{i\leq C}$ corresponds to a unique vector
$\vec{c}$ (determined via the first relation in
Equation~\ref{eq:relations NCc}). However, given a vector $\vec{c}$
one can determine the sequence $(n_{i})_{i\leq C}$ only up to
permutations of the species identities. More information is
intrinsically contained in $(n_{i})_{i\leq C}$ than in $\vec{c}$.

As mentioned, we will also be interested in the statistics of the
population diversity, as described by \textit{e.g.}, Shannon's entropy
$H$ and Simpson's diversity index $S$.  These quantities can be
defined using either the particle-count or the species-count
representations:

\begin{equation}
H=-\sum_{i=1}^{C}\frac{n_{i}}{N}\log\left[\frac{n_{i}}{N}\right]=
-\sum_{k\geq1}c_{k}\frac{k}{N}\log\left[\frac{k}{N}\right]\quad\text{and}\quad
S=1-\sum_{i=1}^{C}\left(\frac{n_{i}}{N}\right)^{2}=1-
\sum_{k\geq1}c_{k}\left(\frac{k}{N}\right)^{2}.
\label{eq:def diversity indices}
\end{equation}
While many variants of Simpson's diversity exist, we have chosen the
``Gini-Simpson'' index \citep{OIKOS} with replacement, also known as
the probability of interspecific encounter \cite{PIE1971},
Gini-Martin, or Blau indices \citep{BLAU}, so that more diverse
populations have a higher value of $S$. Our choice also allows for
analytic derivations not available for other
diversity indices.

We shall analytically derive, whenever possible, probability
distributions over all the quantities defined above. Our results will
be limited to distributions in steady state. Henceforth, we will
define probabilities of a quantity $X$ having a value $x$ as $P(X=x)$,
but we will interchangeably also use the imprecise notation $P(X)$
when no ambiguity exists.

After determining results in steady state for each of the three
\added{neutral} BDI models (\added{sBDI, BDIM, and BDICC}) we will
also analyze the asymptotically large immigration limit. \added{This
  regime may relevant to applications where the per-individual
  immigration rate is higher than their birth and death rates, such as
  in the case of lymphocyte production and maintenance.}  In
particular, we will assume the immigration rate $\alpha$ that defines
the Poisson point process described above to be proportional to a
scaling factor $\Omega$ (\textit{i.e.}, the immigration rate $\alpha
\equiv \widetilde{\alpha}\Omega$ with $\widetilde{\alpha}$ being a
proportionality constant) and then study the $\Omega\to\infty$
limit. We will show that as $\Omega$ increases, the above quantities
also diverge, however, their scaled values
\[
N/\Omega,\quad C/\Omega\quad\text{and}\quad\left(c_{k}/\Omega\right)_{k\geq1}
\]
will be shown to converge in distribution. For example, the convergence in
distribution of $N/\Omega$ to a given constant limit $\ell$ will be
denoted $N/\Omega\xrightarrow[\Omega\rightarrow\infty]{{\cal D}}\ell$
and, when $\ell$ is a constant, the convergence can be characterized by
\begin{equation}
\text{for any }\delta >0,\quad\lim_{\Omega\rightarrow\infty}
P(\left|N/\Omega-\ell\right|<\delta)=1
\label{eq:def convergence distrib}
\end{equation}
\added{(for a more general definition of the convergence in distribution 
where $\ell$ is an arbitrary random variable, see \citet[Chapter 5]{Billingsley2012}).}
This type of convergence implies that
\[
\E{N/\Omega}\xrightarrow[\Omega\rightarrow\infty]{}\ell
\quad\text{and}\quad\var{N/\Omega}\xrightarrow[\Omega\rightarrow\infty]{}0.
\]

\section{\label{sec:mod1}Simple Birth-Death-Immigration Model (sBDI)}

We start with the neutral and independent simple
Birth-Death-Immigration model (sBDI) where individuals are assumed to
be identical, subject to the same birth, death and immigration rates
(neutral), and where the dynamics of each individual is independent of
that of others (independent).  Mutations are not included. One of the
most immediate applications of this sBDI model is within the study of
island biodiversity \citep{Volkov2003,Hubbell2001} where individuals
follow classical birth and death processes, and new species are
introduced to the island via immigration. The ensuing species
abundances are then determined. The main ingredients of the sBDI model
are depicted in Figure~\ref{Fig1}(a) and include (i) immigration, in
which an individual of a new species is added to the system at rate
$\alpha$, (ii) birth, in which each individual gives birth to an
offspring of the same species at rate $r$, and (iii) death, where each
individual dies and is removed from the system at rate $\mu$.

\subsection{\label{sec:BDIanalyze}Derivation of steady state statistics}

We now determine the steady-state probability distribution $P(N)$
of the total number of individuals $N$ in the simple BDI model and
the full probability distribution $P(\vec{c})\equiv P(c_{1},\dots,c_{k},\dots)$
at steady-state. This quantity will lead us to the derivation of the
marginal steady-state probability distributions $P(c_{k})$ and $P(n_{i})$
in the individuals and species count representations, respectively.
From $P(\vec{c})$ we will also be able to obtain the probability
distribution $P(C)$ of the total number of species $C$ at steady-state.
Finally, Shannon's entropy and Simpson's diversity index will be calculated.

The total number of particles $N$ is a random variable that follows a
birth and death process with non-constant rates $\alpha+rN$ and $\mu N$,
respectively. The properties of this birth and death process are known
(see for instance \citet{Bansaye2015a}); in particular for a finite
steady state to exist the condition $r<\mu$ must hold. This constraint
implies that death dominates reproduction so that the number
of individuals $N$ does not diverge at long times. At steady state,
and for $r<\mu$, detail balance leads to the following condition

\begin{equation}
\mu N P(N) = (\alpha+(N-1)r)P(N-1).\label{eq:mod1 balance_N}
\end{equation}

\noindent This equation can be solved iteratively to yield

\begin{equation}
P(N) =
\begin{cases} \displaystyle 
\left(1-{r\over \mu}\right)^{\alpha/r},\quad N=0\\
 \displaystyle  P(0)\left(\frac{r}{\mu}\right)^{N} 
\frac{1}{N!}\prod_{k=0}^{N-1}\left(\frac{\alpha}{r}+k\right),\quad N\geq 1
\end{cases}
\label{eq:mod1 PN}
\end{equation}
%
%
which we recognize as a negative binomial distribution with 
parameters $\alpha/r$ and $r/\mu$, and mean and variance

\begin{equation}
\E{N} = {\alpha/\mu \over 1 - r/\mu}, \quad \var{N} = \frac{\alpha/\mu}{(1-r/\mu)^{2}}.
\label{eq:meanN}
\end{equation}

Equation~(\ref{eq:mod1 PN}) does not resolve how the subpopulations
are distributed within the different species. To determine this
distribution we must derive the distribution $P(\vec{c})$ over the
species-count vector $\vec{c}=(c_{1},\ldots,c_{k},\ldots)$ by
explicitly writing down all possible BDI events and their relative
rates:

\begin{align*}
 &  & (c_{1},c_{2},\ldots) & \xrightarrow{\alpha}(c_{1}+1,c_{2},\ldots) &  & 
\text{Immigration}\\
\text{for }k\geq1\quad &  & (c_{1},\ldots,c_{k},c_{k+1},\ldots) & 
\xrightarrow{rkc_{k}}(c_{1},\ldots,c_{k}-1,c_{k+1}+1,\ldots) &  & \text{Birth}\\
\begin{aligned}\text{for }k\geq2\quad\\
\\
\end{aligned}
 &  & \begin{aligned}(c_{1},\ldots,c_{k-1},c_{k},\ldots)\\
(c_{1},c_{2},\ldots)
\end{aligned}
 & \begin{aligned} & \xrightarrow{\mu kc_{k}}(c_{1},\ldots,c_{k-1}+1,c_{k}-1,\ldots)\\
 & \xrightarrow{\mu c_{1}}(c_{1}-1,c_{2},\ldots)
\end{aligned}
 & \left.\vphantom{\begin{array}{c}
(c_{1})\\
(c_{1})
\end{array}}\right\} \, & \text{Death}
\end{align*}
Since each clone is populated by $k$ individuals the total clone
population is $kc_{k}$, within which each cell can duplicate or die
with rate $r$ or $\mu$. The overall birth and death rates for all
clones of size \added{\protect{$k$}} are thus given by $rkc_{k}$ and $\mu kc_{k}$,
respectively. We can thus write for every $k\geq2$,
\begin{equation}
\left(k-1\right)c_{k-1}\mu
P(c_{1},\ldots,c_{k-1},c_{k},\ldots)=k\left(c_{k}-1\right)r P(c_{1},
\ldots,c_{k-1}+1,c_{k}-1,\ldots),
\label{eq:mod1 balance_ck}
\end{equation}
and for $k=1$,
\begin{equation}
 \mu c_{1} P(c_{1},c_{2},\ldots) = \alpha P(c_{1}-1,c_{2},\ldots).
\label{eq:mod1 balance_ck-1}
\end{equation}
As shown in Appendix~\ref{subsec:S-mod3-Distrib N} \added{for the more general case of 
the BDICC model}, by recursion of
Equation~(\ref{eq:mod1 balance_ck-1}) and using Equation~(\ref{eq:mod1
  balance_ck}), we find

\begin{equation}
P(c_{1},\dots,c_{k},\dots)=P(0,0,\ldots)\left(\frac{\alpha}{r}\right)^{C}
\left(\frac{r}{\mu}\right)^{N} \frac{1}{\prod_{i=1}^{\infty}i^{c_{i}}c_{i}!},
\label{eq:mod1 P vec c}
\end{equation}
with $C=\sum_{k\geq1}c_{k}$ and $N=\sum_{k\geq1}kc_{k}$ as defined
in Equation~(\ref{eq:relations NCc}). The prefactor $P(0,0,\ldots)$
is simply the normalization constant and can be computed as
\begin{align}
P(0,0,\ldots)^{-1} & =\sum_{\vec{c}}\left(\frac{\alpha}{r}\right)^{C}
\left(\frac{r}{\mu}\right)^{N} \frac{1}{\prod_{i=1}^{\infty}i^{c_{i}}c_{i}!}
=\exp\left(\frac{\alpha}{r}\sum_{i=1}^{\infty}\frac{1}{i}
\left(\frac{r}{\mu}\right)^{i}\right)=\left(1-\frac{r}{\mu}\right)^{-\alpha/r},
\label{eq:mod1 P0}
\end{align}
so that finally
\begin{equation}
P(\vec{c})=P(c_{1},\dots,c_{k},\dots)=
\left(1-\frac{r}{\mu}\right)^{\alpha/r}\left(\frac{\alpha}{r}\right)^{C}
\left(\frac{r}{\mu}\right)^{N} \frac{1}{\prod_{i=1}^{\infty}i^{c_{i}}c_{i}!}.
\label{eq:mod1 P vec c-1}
\end{equation}
Note that $P(0,0,\ldots)$ as expressed in Equation~(\ref{eq:mod1 P0})
corresponds to the $N=0$ case in Equation~(\ref{eq:mod1 PN}) since
the state with no individuals present in the population can
only be represented by the configuration $\vec{c}=(0,0,\ldots)$.

We can now use Equation~(\ref{eq:mod1 P vec c}) to determine the
distribution for the total number of species $C$. To do this, we
consider its moment generating function $M_{C}(\xi)$ defined as the
average $\E{\exp\left(\xi C\right)}$
\[
M_{C}(\xi)\equiv\E{\exp\left(\xi C\right)}=\sum_{c_{1},\dots,c_{k},\dots}
e^{\xi C}\left(1-\frac{r}{\mu}\right)^{\alpha/r}\left(\frac{\alpha}{r}\right)^{C}
\left(\frac{r}{\mu}\right)^{N} \frac{1}{\prod_{i=1}^{\infty}i^{c_{i}}c_{i}!},
\]
with $C=\sum_{k\geq1}c_{k}$ and $N=\sum_{k\geq1}kc_{k}$. We find

\begin{equation}
M_{C}(\xi)=\left(1-\frac{r}{\mu}\right)^{(1-e^{\xi})\alpha/r} 
\sum_{c_{1},\dots,c_{k},\dots}\left(1-\frac{r}{\mu}\right)^{\alpha e^{\xi}/r}
\left(\frac{\alpha e^{\xi}}{r}\right)^{C}\left(\frac{r}{\mu}\right)^{N} 
\frac{1}{\prod_{i=1}^{\infty}i^{c_{i}}c_{i}!}.
\label{eq:mod1 M_C}
\end{equation}
Upon comparing Equation\,(\ref{eq:mod1 P vec c-1}) with the terms in
the last summation in Equation~(\ref{eq:mod1 M_C}) we can easily see
that the terms within the sum define the probability
$P(c_{1},\dots,c_{k},\dots)$ of \deleted{an} another simple independent BDI
model with immigration rate $\alpha\to\alpha e^{\xi}$. Thus, from
normalization, the sum in Equation\,(\ref{eq:mod1 M_C}) is equal to
one. By writing $M_{C}(\xi)$ in the form
\[
M_{C}(\xi)=\exp\left[\left(e^{\xi}-1\right)
\frac{\alpha}{r}\log\left({1\over 1 - r/\mu}\right)\right], 
\]
which is a moment generating function of a Poisson
random variable with parameter
$(\alpha/r)\log\left[1/\left(1-r/\mu\right)\right]$ (see
\citet[Chapter 4]{Grimmett2001}) \added{we find}
\begin{equation}
P(C)=\left(1-\frac{r}{\mu}\right)^{\frac{\alpha}{r}}
\frac{\left(\frac{\alpha}{r}\log\left[1/\left(1-\frac{r}{\mu}\right)\right]\right)^{C}}{C!}.
\label{eq:mod1 P C}
\end{equation}
Using this distribution, we find 

\begin{equation}
\E{C}=\var{C} = \left(\frac{\alpha}{r}\right)\log\left[\frac{1}{1-r/\mu}\right].
\label{EC}
\end{equation}

We now find the marginal probability $P(c_{k})$ for the number of
species $c_{k}$ with $k$ individuals regardless of all others. By
using Equation~(\ref{eq:mod1 P vec c-1}), separating out the $c_{k}$
terms, we obtain

\begin{align}
P(c_{k}) & =\sum_{(c_{i})_{i\neq k}}P(c_{1},c_{2},\ldots,c_{k-1},c_{k},c_{k+1})\label{eq:mod1 Pc_k}\\
 & =\left(1-\frac{r}{\mu}\right)^{\alpha/r}\sum_{(c_{i})_{i\neq k}}\prod_{j=1}^{\infty}
\frac{1}{c_{j}!}\left(\frac{1}{j}\,\frac{\alpha}{r}\,\left(\frac{r}{\mu}\right)^{j}\right)^{c_{j}}\nonumber \\
 & =\frac{1}{c_{k}!}\left(\frac{1}{k} \frac{\alpha}{r} \left(\frac{r}{\mu}\right)^{k}\right)^{c_{k}}
\left(1-\frac{r}{\mu}\right)^{\alpha/r}\prod_{j\neq k}^{\infty}\exp
\left(\frac{1}{j}\,\frac{\alpha}{r}\,\left(\frac{r}{\mu}\right)^{j}\right)\nonumber \\
 & =\frac{1}{c_{k}!}\left(\frac{1}{k}\frac{\alpha}{r}
\left(\frac{r}{\mu}\right)^{k}\right)^{c_{k}}\exp\left(-\frac{1}{k}\,\frac{\alpha}{r}\,
\left(\frac{r}{\mu}\right)^{k}\right),\nonumber
\end{align}
which is a Poisson distribution with parameter equal to the mean and variance

\begin{equation}
\E{c_{k}} = \var{c_{k}} =\frac{\alpha}{r}\frac{1}{k}\left(\frac{r}{\mu}\right)^{k}.
\label{Eck}
\end{equation}

Next, we determine the marginal distribution of the number $n_{i}$
of individuals belonging to species $i$. By taking the mean of the
first relation in Equation~(\ref{eq:relations NCc}), we find
\[
\E{c_{k}}=\E{\sum_{i=1}^{C}\ind{n_{i},k}}=\E{\sum_{i=1}^{C}\E{\ind{n_{i},k}|C}}.
\]
Since species are assumed to be non-interacting, the random variables
$\left(n_{i}\right)_{i\leq C}$ are independent and identically distributed (iid)
and are also independent of $C$. Thus, for every $1\leq i\leq C$
we can write
\[
\E{\ind{n_{i},k}|C}=\E{\ind{n_{1},k}}=P(n_{1}=k)
\]
for which $P(n_{1}=k)$ is still undetermined. The above relation yields
\begin{equation}
\E{c_{k}}=\E{C}P(n_{1}=k), 
\label{eq:mod1 Eck}
\end{equation}
in which $\E{c_{k}}$ and $\E{C}$ are determined by
Equations~(\ref{Eck}) and (\ref{EC}). Since all $n_{i}$ values are
identically distributed and $P(n_{i}=k)=P(n_{1}=k)$, we can finally write
$P(n_{i})$ for any species $i$:

\begin{equation}
P(n_{i})=\frac{1}{n_{i}}\left(\frac{r}{\mu}\right)^{n_{i}}
\frac{-1}{\log\left[1-r/\mu\right]}.
\label{eq:mod1 Pn_i}
\end{equation}
Thus, every $n_{i}$ follows Fisher's logarithmic series distribution
\citep{Fisher1943} \added{with} parameter $r/\mu$.  Note that although the
distribution $P(N)$ for the total population $N$ depends on the
immigration rate, the distribution in Equation~(\ref{eq:mod1 Pn_i}) is
\textit{independent} of $\alpha$. This is because each immigration
event necessarily introduces a new species but does not influence the
dynamics of a species already present. Once introduced, the evolution
of any species depends solely on its birth and death rates $r$ and
$\mu$.

Finally, we can use Equation~(\ref{eq:mod1 P vec c-1}) to determine
the expected Shannon's entropy and Simpson's diversity index, as
defined in Equation~(\ref{eq:def diversity indices}). Using a similar
procedure to the one used to determine $P(c_{k})$, we isolate the
$c_{k}$ term in the definition of $P(\vec{c})$ and find the same form
in terms of $c_{k}-1$. \added{Note that} we can write
the mean of $c_{k}f(N)$, for any function $f(N)$, as
\[
\E{c_{k}f(N)}=\frac{1}{k}\frac{\alpha}{r}\left(\frac{r}{\mu}\right)^{k}\E{f(N+k)}.
\]
By considering the functions $f(x)=\log\left(x/k\right)/x$ and $f(x)=(k/x)^{2}$
we find the respective expressions for Shannon's Entropy and Simpson's
diversity index

\begin{equation}
\E H=\frac{\alpha}{r}\sum_{k=1}^{\infty}\left(\frac{r}{\mu}\right)^{k}
\E{\frac{\log\left(\frac{N+k}{k}\right)}{N+k}}\quad\text{and}\quad
\E S=1-\frac{\alpha}{r}\sum_{k=1}^{\infty}k\left(\frac{r}{\mu}\right)^{k}
\E{\left(\frac{1}{N+k}\right)^{2}}.\label{eq:mod1 EH ES}
\end{equation}
Since the distribution of $N$ is known and given by Equation~(\ref{eq:mod1 PN}),
we can use Equation~(\ref{eq:mod1 EH ES}) to numerically compute
$\E H$ and $\E S$.

\added{For completeness, we also derive results for the sBDI process with a
finite number of clones $Q$ that each carry a finite immigration rate
into the system. In Appendix \ref{FINITENUMBER}, we use the detailed
balance conditions to derive explicit steady state probability
distributions over the particle count vector $n_{i}$, the species
count vector $c_{k}$, and the number of clones in the sample $C$.}

\subsection{\label{subsec:LIL}Fast immigration limit}

We now consider the \added{large} immigration limit of
the sBDI model.  While at steady-state the distribution of the number
of individuals in each species $P(n_{i})$, given by
Equation~(\ref{eq:mod1 Pn_i}), is independent of \deleted{the} $\alpha$, the
distributions $P(N)$, $P(C)$ and $P(c_{k})$ do depend on the total
immigration rate $\alpha$ as indicated in Equations~(\ref{eq:mod1
  PN}),~(\ref{eq:mod1 P C}) and~(\ref{eq:mod1 Pc_k}),
respectively. Since immigration always introduces a new species, the
\textit{per clone} immigration rate is zero.  To study the large
immigration regime in which each clone has a finite immigration rate,
we assume $\alpha \equiv \widetilde{\alpha}\,\Omega$ scales as the
parameter $\Omega \to \infty$, which can be thought of as the total
number of different clones that can immigrate into the system per unit
time.  Increasing $\alpha$ will introduce new individuals and new
species to the system, so one can intuitively conclude that the total
population $N$ and number of species $C$, as well as the number of
species with $k$ individuals $c_{k}$, will also increase. We also show
that, as $\Omega$ increases, the scaled values $N/\Omega$, $C/\Omega$
and $c_{k}/\Omega$ converge in distribution to a constant, as
described in Equation~(\ref{eq:def convergence distrib}), with average
values given by
\[
\frac{N}{\Omega}\xrightarrow[\Omega\rightarrow\infty]{{\cal D}}
\frac{\widetilde{\alpha}}{\mu-r},\quad
\frac{C}{\Omega}\xrightarrow[\Omega\rightarrow\infty]{{\cal D}}
\frac{\widetilde{\alpha}}{r}\log\left[\frac{1}{1-r/\mu}\right]\quad
\text{and}\quad\frac{c_{k}}{\Omega}\xrightarrow[\Omega\rightarrow\infty]{{\cal D}}
\frac{\widetilde{\alpha}}{r}\,\frac{(r/\mu)^{k}}{k},
\]
and vanishing variances. A rigorous proof is given in
Appendix~\ref{subsec:S-mod1-Convergence high immigration}.  
Moreover, we can also write the convergence in distribution of the
scaled Shannon's entropy $H/\log\Omega$ and Simpson's diversity index
$S$,

\[
{H\over \log\Omega}\xrightarrow[\Omega\rightarrow\infty]{{\cal D}}1\quad
\text{and}\quad S\xrightarrow[\Omega\rightarrow\infty]{{\cal D}}1,
\]
as also derived in Appendix~\ref{subsec:S-mod1-Convergence high
  immigration}. We can now use
the scaling results above to infer that
\[\frac{\E{H}}{\log \Omega}= 1-{\cal O}(\Omega^{-1}), \quad 
\E{S} = 1 - {\cal O}(\Omega^{-1}).
\]

\begin{figure}[h]
\begin{center}
\includegraphics[width=7in]{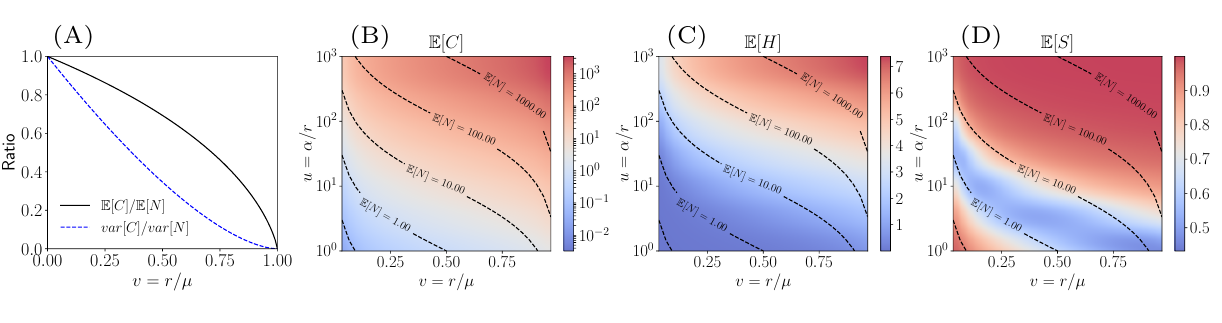}\vspace{-5mm}
\caption{(a) The ratios $\protect\E C/\protect\E N$ and $\protect\var
  C/\protect\var N$ as a function of $v = r/\mu$. This ratio is
  independent of the immigration rate $\alpha$ and can be explicitly
  determined by our theoretical results.  (b), (c) and (d) Analytical
  results for the total richness $\protect\E C$, Shannon's entropy
  $\protect\E H$ and Simpson's index $\protect\E S$, respectively, as
  functions of $u=\alpha/r$ and $v=r/\mu$. Both $\protect\E C$ and
  $\protect\E H$ increase with $\alpha/r$, $r/\mu$, and $\protect\E N$.
  Simpson's index $\protect\E S$ has a similar behavior only if
  $\protect\E N= u v/(1-v)$ is greater than $\sim 10$.
\label{Fig2}}
\end{center}
\end{figure}

\subsection{Interpretation of results}

All distributions computed above depend on two nondimensional quantities:
$u\equiv \alpha/r$ and $v^{-1}\equiv \mu/r$ (with $0 \leq v<1$). 
%
%
Also note that the ratios $\E C/\E N$ and $\var C/\var N$ are
``immigration-invariant'' in the sense that they depend only on
$v=r/\mu$ and are independent of the immigration rate $\alpha$ and/or
$u$.  These ratios are plotted in Figure~\ref{Fig2}(a).
For systems where the immigration rate is difficult to estimate, such
quantities can be useful probes of the system. For example, $\E C/\E
N$ represents the relative abundance of species with respect to the
total number of individuals in the system.  By construction, this
ratio cannot exceed unity.  The limiting case of $\E C/\E N \to 1$
corresponds to a completely heterogeneous system where every individual
belongs to a different species while $\E C/\E N \to 0$ corresponds to
homogeneous systems where the entire population is dominated by very
few species.

In Figures~\ref{Fig2}(b), (c), and (d) we show the respective heat-maps of $\E C$,
$\E H$ and $\E S$ as a function of $u=\alpha/r$ and $v=\alpha/r$,
along with level-sets of $\E N$.  These figures have been generated
using Equations~(\ref{eq:mod1 P C}) and~(\ref{eq:mod1 EH ES}) and show
that both $\E C$ and $\E H$ increase with $u$, $v$, and $\E N$.  In
our plots, $\E S$ is not strictly monotonic despite the expectation
that $\E S$, as a measure of diversity, would follow the same trend as
$\E C$ and $\E{H}$. However, this qualitative discrepancy occurs only
in the small $\E N$ regime where there is a high probability that
there \added{are} no individuals in the system and diversity loses its
meaning. In the extreme limit of $N\to 0$, we find $C\to 0$ and $H\to
0$, but $S \to 1$, giving rise to the nonmonotonic pattern for
$\E{S}$.

\section{\label{sec:mod-2}Birth-Death-Immigration Model with Mutation (BDIM)}

In this section we consider a Birth-Death-Immigration Model with
Mutation (BDIM).  Mutation events are particularly relevant in ecology
as they lead to speciation within populations \citep{Lambert2011}, and
in studies of gene domain family evolution \citep{Karev2002}.  In the
BDIM process, we still assume individuals and species are
non-interacting and that birth, death, immigration, and mutation rates
do not depend on the state of the system. We allow an individual of a
given species to mutate and give rise to a new, yet unrepresented,
species. Mutations are assumed to be neutral in that an individual
arising from mutation maintains the same birth and death rates as the
rest of the population.

We start by allowing mutations only in offspring arising immediately
after their birth, as illustrated in Figure~\ref{Fig1}(b). For each
birth event there is a probability $\epsilon$ (with $0 \leq \epsilon
\leq 1$) that the offspring is mutated, representing a completely new
species. This mechanism is applicable to \textit{e.g.}, bacterial
populations where DNA replication can induce a gene mutation that will
be carried by the newborn cell.  The subsequent theoretical analysis
will be carried out within the framework of a single mutation at birth
as described here. However our mathematical treatment is not limited
to this specific case and, in
Subsection~\ref{subsec:Other-mutation-mechanisms}, we will apply the
same tools to study other relevant scenarios such as ``somatic''
mutations that can occur any time during the lifetime of an
individual, \added{and} ``double'' (or ``symmetric'')
mutations where both parent and offspring mutate upon birth.

\subsection{\label{sec:mod-2analysis}Derivation of steady state statistics}

In Section~\ref{sec:mod1}, we were able to use reversibility and
detailed balance to determine $P(\vec{c})$, the probability for a
given species-count configuration $\vec{c}=(c_{1},\ldots,c_{k},\dots)$
to occur. The introduction of mutation, however, makes the system
irreversible and analytically evaluating $P(\vec{c})$ becomes
prohibitively complex. We can nonetheless exploit some general features of
the BDIM model, such as neutrality and independence, to 
extract results such as the mean and the variance of $C$ and
$c_{k}$. The evaluation of other quantities such as the mean of the
diversity indices $\E H$ and $\E S$ will require numerical
simulations. Our theoretical analysis relies on two important features
of the model:

\begin{itemize}
\item Because mutations do not affect overall birth or death rates but
  only the species to which newborns belong, the distribution for
  $P(N)$ remains identical to the one derived in
  Equation~(\ref{eq:mod1 PN}) for the simple BDI model. Hence, in the
  BDIM model, the overall growth rate due to immigration and birth is
  still $\alpha+Nr$ and the overall death rate is still $\mu N$. The
  resulting $P(N)$ is independent of mutation events and
  Equation~(\ref{eq:mod1 PN}) still holds.

\item The marginal distribution $P(n_{i})$ of the number of
  particles $n_{i}$ of species $i$ still follows a logarithmic series
  distribution as in Equation\,(\ref{eq:mod1 Pn_i}), but with the
  replacement $r\to r(1-\epsilon)$.  Intuitively, this can be
  understood by noting that under mutation a new individual is
  introduced into the $n_{i}$ population with rate $r(1-\epsilon)$
  instead of $r$, since the ``remainder'' $r\epsilon$ is the rate at
  which a new individual in a new species arises. The dynamics of the
  $n_{i}$ individuals thus remains unchanged, provided the birth rate
  is modified to $r(1-\epsilon)$ to account for the diminished births
  within the given species. In Appendix~\ref{subsec:S-mod2 distrib
    n_i} we provide a more rigorous justification of this
  argument. Also, in Figure~\ref{fig:S-mod2: Sim distrib n_i} of
  Appendix~\ref{subsec:S-mod2 distrib n_i} we plot the probability
  distribution for the number of individuals in a given species as
  determined from simulations of the BDIM model, compare our findings
  to the expected logarithmic distribution, and show good agreement
  between the two.  Thus, both theoretically and numerically, we
  verify that $P(n_{i})$ follows a logarithmic series distribution
  \added{with} parameter $r\left(1-\epsilon\right)/\mu$ for all values
  of $0\leq \epsilon \leq1$:

\begin{equation}
P(n_{i})=\frac{1}{n_{i}}\left(\frac{r(1-\epsilon)}{\mu}\right)^{n_{i}}
\frac{-1}{\log\left[1-r(1-\epsilon)/\mu\right]}.
\label{eq:mod2 Pni}
\end{equation}

\end{itemize}
Once the $P(n_{i})$ and $P(N)$ distributions are known for the BDIM
model we can use Equation~(\ref{eq:relations NCc}) and the fact that
the $(n_{i})_{i\leq C}$ are iid and independent of $C$ to express the
mean of the third relation of Equation~(\ref{eq:relations NCc}) as
\[
\E N=\E{\sum_{i=1}^{C}\E{n_{i}|C}}=\E C\E{n_{1}},
\]
so that
\begin{equation}
\E{C}=\frac{\alpha/\mu}{1-r/\mu}\log\left(\frac{1}{1-r(1-\epsilon)/\mu}\right)
\frac{1-r(1-\epsilon)/\mu}{r(1-\epsilon)/\mu}.
\label{eq:mod2 EC}
\end{equation}
Using the moment generating function of $N$, we can similarly
determine the variance of $C$ as shown in detail in
Appendix~\ref{subsec:S-mod2 Moments of C}
\begin{equation}
\var C=\E C\left[\frac{\E C}{(\alpha/r)}+
\log\left(1-\frac{r(1-\epsilon)}{\mu}\right)+1\right].\label{eq:mod2 VC}
\end{equation}
Finally, we take the mean of the first relation in
Equation~(\ref{eq:relations NCc}).  Since all the $(n_{i})_{i\leq C}$
are iid and independent of $C$ we can write

\begin{equation}
\E{c_{k}}=\E{C} P(n_{1}=k)=\frac{\alpha/\mu}{1-r/\mu}\frac{1}{k}
\left(1-\frac{r(1-\epsilon)}{\mu}\right)
\left(\frac{r(1-\epsilon)}{\mu}\right)^{k-1}.
\label{eq:mod2 Eck}
\end{equation}
For the variance of $c_{k}$, we also use the definition in
the first relation in
Equation~(\ref{eq:relations NCc}) to find
\[
c_{k}c_{\ell}= \ind{k,\ell}\sum_{i=1}^{C}\ind{n_{i},k} + 
\sum_{i=1}^{C}\sum_{j\neq i}\ind{n_{i},k}\ind{n_{j},k}.
\]
Upon using Equation~(\ref{eq:mod2 Pni}) to take the mean of this
expression, and recalling that $n_{i}$ is independent of $n_{j\neq i}$
and $C$, we find
\begin{align}
\E{c_{k}c_{\ell}} & = \ind{k,\ell}\E{c_{k}}+\E{C(C-1)}P(n_{i}=k)P(n_{i}=\ell), \nonumber \\
\var{c_{k}} & = \E{c_{k}}+\frac{\var C-\E C}{\E C^{2}}\,\E{c_{k}}^{2}.
\label{eq:mod2 Vck}
\end{align}
These expressions are also valid for the sBDI model, but since 
$C$ and $c_{k}$ are Poisson-distributed in that case, $\var{C}=\E{C}$ and 
$\var{c_{k}}=\E{c_{k}}$ (see below).

We can use Equations~(\ref{eq:mod2 Pni}), (\ref{eq:mod2 Eck}),
(\ref{eq:mod2 EC}), and Appendix (\ref{subsec:S-mod2 Moments of C}) to further
develop the second moments. For example,
\[
\E{c_{k}c_{\ell\neq k}}=\left(\frac{r(1-\epsilon)}{\mu}\right)^{k \ell}
\frac{1}{k\ell}\frac{
\left[\frac{\alpha}{\mu}\left(1-\frac{r(1-\epsilon)}{\mu}\right)
+\epsilon\frac{r}{\mu}\right]
\left[\frac{\alpha}{\mu}\left(1-\frac{r(1-\epsilon)}{\mu}\right)
\right]}{\left(1-\frac{r}{\mu}\right)^{2}\left(\frac{r(1-\epsilon)}{\mu}\right)^{2}},
\]
which reduces to the simple BDI result $\E{c_{k}}^{2}$  when $\epsilon = 0$.
For $k=\ell$, we have 

\begin{equation*}
\var{c_{k}} = \E{c_{k}}+\epsilon\frac{(\alpha/\mu)(r/\mu)}{(1-r/\mu)^{2}k^{2}}
\left(\frac{r(1-\epsilon)}{\mu}\right)^{2(k-1)}.
\end{equation*}

\subsection{Fast immigration limit}

We now study the large immigration limit of the BDIM model. As done in
Section\,\ref{subsec:LIL} we set $\alpha=\widetilde{\alpha}\,\Omega$
and consider the $\Omega\to\infty$ limit. Since the dynamics of
$N/\Omega$ remain unchanged in the BDIM model compared to the
dynamics in the simple BDI model, we recover convergence in
distribution for $N/\Omega$ towards the constant
$\widetilde{\alpha}/(\mu-r)$ as $\Omega\to\infty$. Following the same
procedures illustrated in Appendix~\ref{subsec:S-mod1-Convergence high
  immigration} for the simple BDI model, and using the moment
generating functions of $C$ and $c_{k}$ we can also prove the
convergence in distribution of $C/\Omega$ and $c_{k}/\Omega$ towards
the following

\begin{align*}
\frac{C}{\Omega} & \xrightarrow[\Omega\rightarrow\infty]{{\cal D}}
\frac{(\widetilde{\alpha}/\mu)}{1-r/\mu}\frac{1-r(1-\epsilon)/\mu}{r(1-\epsilon)/\mu}
\log\left(\frac{1}{1-r(1-\epsilon)/\mu}\right),\\
\frac{c_{k}}{\Omega} & \xrightarrow[\Omega\rightarrow\infty]{{\cal D}}
\frac{(\widetilde{\alpha}/\mu)}{1-r/\mu}\frac{1-r(1-\epsilon)/\mu}{k}
\left(\frac{r(1-\epsilon)}{\mu}\right)^{k-1}.
\end{align*}
Finally, the convergence in distribution of the scaled Shannon's entropy
$H/\log\Omega$ and Simpson's diversity index $S$ are

\[
\frac{H}{\log\Omega}\xrightarrow[\Omega\rightarrow\infty]{{\cal D}}1\quad
\text{and}\quad S\xrightarrow[\Omega\rightarrow\infty]{{\cal D}}1.
\]

\subsection{Interpretation of results and comparison with sBDI model}

\begin{figure}
\begin{center}\includegraphics[width=6.8in]{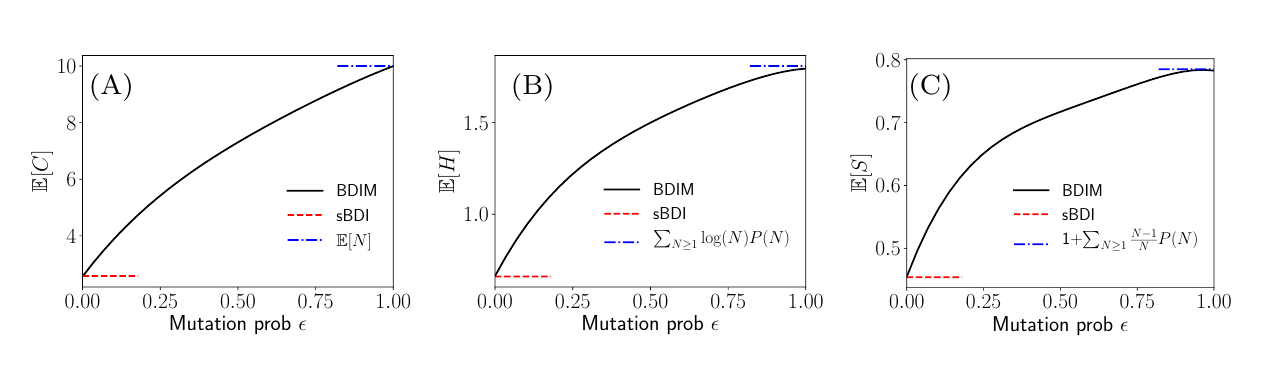}\vspace{-5mm}
\caption{\label{Fig4}Comparison of various diversity
  indices between the sBDI model and the BDIM model at varying values
  of mutation probability $\epsilon$. For both models, we set
  $\alpha=1$, $\mu=1$, and $r=0.9$. (a) Total richness $C$ as determined
  using Equations~(\ref{eq:mod2 EC}) and~(\ref{eq:mod2 VC}). (b)
  Expected Shannon's entropy $\E{H}$, determined with simulations of
  the model with mutation at division. (c) Expected Simpson's
  diversity index $\E{S}$, determined with simulations of the model
  with mutation at division. In each case, mutation increases
  diversity (black curves) relative to that of the sBDI model (blue
  dot-dashed lines). The maximum diversity is obtained in the limit where
  mutation always occurs (i.e. $\epsilon=1$) and all individuals
  belong to different species such that $C=N$ (red dashed lines).}
\end{center}
\end{figure}

In this subsection we compare features of the sBDI and BDIM models.
Note that only three parameters are necessary to characterize all
results obtained in both models: $u\equiv\alpha/r \geq 0$, $v\equiv
r/\mu$ $(0\leq v \leq 1)$, and $0\leq \epsilon \leq 1$. In
Figures~\ref{Fig4}(a), (b) and (c) we plot the total
richness $C$, Shannon's entropy $H$, and Simpson's index $S$ as
defined in Equation~(\ref{eq:def diversity indices}).  In these
figures, $\E{C}$ was determined using Equations~(\ref{eq:mod2 EC})
and~(\ref{eq:mod2 VC}), while $\E{H}$ and $\E{S}$ were found from
simulations.

The expected diversity indices for the simple BDI model \deleted{is
  independent of \protect{$\epsilon$} and} are shown by the dashed
\added{red} horizontal lines. As can be seen from
Figure~\ref{Fig4}, all measures of diversity in the BDIM model
increase with $\epsilon$ and reach their maximum at $\epsilon = 1$
(shown by the blue dot-dashed lines) when all individuals give rise to
mutant offspring.  Upon setting $\epsilon = 1$, and assuming a nonzero
population $N$, Equation~(\ref{eq:def diversity indices}) yields
$C=N$, $H=\log N$ and $S=(N-1)/N$, mirroring the fact that each
species only has one individual.  By noting that $H=0$ and $S=1$ for
$N=0$, we find \deleted{in general} for $\epsilon=1$

\[
\E C=\E N,\quad\E H=\sum_{N\geq1}\log N\,P(N)\quad\text{and}
\quad\E S=1+\sum_{N\geq1}\frac{N-1}{N}\,P(N),
\]
with $P(N)$ being the probability distribution of $N$ given in Equation~(\ref{eq:mod1 PN}).

For general $\epsilon$ in the BDIM model, since we cannot analytically
determine the probability for the species-count vector $\vec{c}$, we
cannot derive explicit formulae for $\E H$ and $\E S$ as we did for
the simple BDI model in Section~\ref{sec:mod1}.  However, we can
estimate both $\E H$ and $\E S$ by approximating $c_k$ with
$\E{c_{k}}$ in Equation~(\ref{eq:def diversity indices}) \added{to} find

\begin{equation}
\E H\simeq-\left(1-\frac{r(1-\epsilon)}{\mu}\right)\sum_{k\geq1}
\left(\frac{r(1-\epsilon)}{\mu}\right)^{k-1}\log\left[\frac{\mu}{\alpha}
k\left(1-\frac{r}{\mu}\right)\right]\quad\text{and}\quad
\E S\simeq1-\frac{\mu}{\alpha}\frac{1-r/\mu}{1-r(1-\epsilon)/\mu}.
\label{eq:mod2 EH ES}
\end{equation}
We compare these approximations with results obtained from numerical
simulations in Figure~\ref{fig:S-mod2: Estim HS} in Appendix B.1. As
can be seen our analytical estimates become more accurate as the
average number of individuals $\E N=\alpha/(\mu-r)$ increases, that
is, for $\alpha\gg\mu$ and $\mu\gtrsim r$.  For $\E N\geq5$, both
estimates for Shannon's entropy and Simpson's diversity index fall
within $10\%$ of their simulated values.

\subsection{\label{subsec:Other-mutation-mechanisms}Alternative mutation mechanisms}

The BDIM model, as described above, assumes that mutations occur with
probability $\epsilon$ during each birth event. We can very easily
adapt the mathematical reasoning used in
Section\,\ref{sec:mod-2analysis} to characterize other types of
mutation processes. Note that if mutation events add more species to
the system, but do not change the overall birth and death rates of the
population, the total-population distribution $P(N)$ will remain
unchanged from the expression found in Equation (\ref{eq:mod1 PN}) for
the simple BDI model. This will be the case for the two alternative 
mutation mechanisms described below.

\begin{description}
\item [{Somatic~mutation:}] Each individual may spontaneously mutate
  at constant rate $\eta>0$ over its lifetime, giving rise to an
  individual of a new species. Such a birth-independent mutation might 
be a reasonable model for \textit{e.g.}, DNA damage 
or epigenetic changes in a cell.
%
%
In this scenario, for a given $n_{i}$ population, new individuals are
added to the same $i$ species at rate $rn_{i}$ and removed at rate
$(\mu+\eta)n_{i}$ since mutation events will effectively transfer an
individual from a given species to a new one. Hence, the distribution
for $P(n_{i})$ should remain a logarithmic series distribution as in
Equation~(\ref{eq:mod1 Pn_i}) but with parameter $r/(\mu+\eta)$. All
theoretical results found in Section~(\ref{sec:mod-2analysis}) remain
the same in this case provided we replace
$\epsilon\to\eta/(\mu+\eta)$.

\item [{Double~mutation:}] Both parent and offspring may spontaneously
  mutate at birth, as for example in symmetric stem cell
  differentiation.  More generally, we can assume that one of the two
  individuals mutates to a new species with probability $\epsilon_{1}$
  and that both mutate into two new species with probability
  $\epsilon_{2}$.  In this case, for a given $n_{i}$ population, new
  individuals are added at rate $r(1-\epsilon_{1}-\epsilon_{2})n_{i}$
  to species $i$ and removed at rate $(\mu+r\epsilon_{2})n_{i}$. The
  number of individuals in species $i$ should thus still be
  logarithmically distributed as in Equation~(\ref{eq:mod1 Pn_i}), but
  with parameter $r(1-\epsilon_{1}-\epsilon_{2})/(\mu+r\epsilon_{2})$.
  All theoretical results found in Section~\ref{sec:mod-2analysis}
  remain the same, provided we replace
  $\epsilon\to(r\epsilon_{2}+\mu(\epsilon_{1}+\epsilon_{2}))/(r\epsilon_{2}+\mu)$.
\end{description}

\section{\label{sec:mod-3}Birth-Death-Immigration Model with Carrying Capacity
(BDICC)}

In the third and final model analyzed in this paper, we include an
important interaction within the total population -- a carrying
capacity that is typically used to represent resource limitations. The
more individuals present in a system, the more they need to share
resources, potentially affecting survival or reproduction rates.  The
carrying-capacity concept is ubiquitous in ecology \added{such as for
  species on an island with finite resources that limit the total
  population.}  Other applications may include lymphatic growth which
is known to be induced by several molecules, in particular cytokines
\citep{Tan2001} that may become insufficient to sustain further
proliferation of T-cells if the population becomes too large.

We first consider a carrying capacity on the death rate of each
individual and derive analytical results; more general cases will be
addressed via numerical simulations. As shown in Figure~\ref{Fig1}(c),
the only difference between our BDI model with carrying capacity
(BDICC) and the sBDI model is that the death rate now depends on the
total number of individuals in the system $N$.  We assume that
$\mu(N)$ is an increasing function with $N$ as dwindling resources led
by population increases will also increase the death rate. It is
important to remark that populations described by the BDICC model do
not evolve independently. Since the dynamics of each individual now
depends on that of all others, there is a global, but ``neutral''
interaction. In contrast to the two previous models, the number of
individuals in each species $\left(n_{i}\right)_{i\leq C}$, can no
longer be considered an independent random variable so that

\[
\E{c_{k}}\neq\E{C} P(n_{1}=k).
\]
The equality of the quantities on the left and right hand sides above
\added{was} used in the previous analysis to determine
Equations~(\ref{eq:mod1 Eck}) and~(\ref{eq:mod2 Eck}) and is no longer
applicable to the BDICC model.

\subsection{\label{subsec:mod-3 mod analysis}Derivation of steady state statistics}

We first consider the dynamics of the total number of individuals $N$
and study how $P(N)$ is modified in the BDICC model. In this case, the
overall population still undergoes a birth and death process with
rates $\alpha+rN$ and $\mu(N)N$, respectively. The properties of birth
and death process with non homogeneous rates are known
\citep{Bansaye2015a}.  In particular, in the case of an increasing
function $\mu(N)>0$, the conditions for the existence of a steady state
is \deleted{$\mu(N) > 0$ and}
\[
\lim_{N\rightarrow\infty}\mu(N)>r.
\]
More general conditions for the existence of a steady-state
configuration have been detailed in the case of a non-increasing death
rate $\mu(N)$ \citep{Bansaye2015a}.  If \added{a} steady-state exists, then
$P(N)$ can be found using detailed balance, similar to what was done
in Section~\ref{sec:mod1}
\begin{equation}
P(N)=\begin{cases}
\displaystyle \frac{1}{Z_{\alpha,r}}, \quad N=0, \\
\displaystyle \frac{1}{Z_{\alpha,r}}\frac{1}{N!}{\displaystyle \prod_{k=0}^{N-1}}
\frac{\alpha+r k}{\mu(k+1)}, \quad N\geq 1
\end{cases}
\label{eq:mod3 PN}
\end{equation}
with $Z_{\alpha,r}$ a normalization constant given by
\[
Z_{\alpha,r}=1+ \sum_{n=1}^{\infty}\frac{1}{n!}{\displaystyle \prod_{k=0}^{n-1}}
\frac{\alpha+rk}{\mu(k+1)}.
\]

\noindent To determine $P(\vec{c})$, $P(C)$ and $P(c_{k})$, we rely on
reversibility of the system and detailed balance.  Interestingly,
while a non-constant death rate $\mu(N)$ preserves detailed balance, a
non-constant growth function $r(N)$ does not strictly obey detailed
balance.  We will come back to this point further in the discussion,
in Section~(\ref{subsec:mod3 CC on birth}).  For now, we consider
$\mu(N)$ and constant $r$ and write all possible transitions of the
system as was done in Section~\ref{sec:BDIanalyze}

\begin{align*}
 &  & (c_{1},c_{2},\ldots) & \xrightarrow{\alpha}(c_{1}+1,c_{2},\ldots) & 
 & \text{Immigration}\\
\text{for }k\geq1\quad &  & (c_{1},\ldots,c_{k},c_{k+1},\ldots) & 
\xrightarrow{rkc_{k}}(c_{1},\ldots,c_{k}-1,c_{k+1}+1,\ldots) &  & \text{Birth}\\
\begin{aligned}\text{for }k\geq2\quad\\
\\
\end{aligned}
 &  & \begin{aligned}(c_{1},\ldots,c_{k-1},c_{k},\ldots)\\
(c_{1},c_{2},\ldots)
\end{aligned}
 & \begin{aligned} & \xrightarrow{\mu(N)kc_{k}}(c_{1},\ldots,c_{k-1}+1,c_{k}-1,\ldots)\\
 & \xrightarrow{\mu(N)c_{1}}(c_{1}-1,c_{2},\ldots)
\end{aligned}
 & \left.\vphantom{\begin{array}{c}
(c_{1})\\
(c_{1})
\end{array}}\right\} \, & \text{Death}
\end{align*}
which differ from the ones written in Section\,\ref{sec:BDIanalyze}
by virtue of $\mu\to\mu(N)$ with $N=\sum_{k\geq1}kc_{k}$. By assuming 
detailed balance, we write

\begin{equation}
\mu(N)kc_{k} P(c_{1},\ldots,c_{k-1},c_{k},\ldots)=
(k-1)\left(c_{k-1}+1\right)r P(c_{1},\ldots,c_{k-1}+1,c_{k}-1,\ldots),
\label{eq:mod3 balance_ck}
\end{equation}
for $k\geq2$, while for $k=1$ the following holds

\begin{equation}
\mu(N) c_{1} P(c_{1},c_{2},\ldots)= \alpha P(c_{1}-1,c_{2},\ldots).
\label{eq:mod3 balance2_ck}
\end{equation}
We follow the same procedure as in Section\,\ref{sec:BDIanalyze}
and iterate Equation~(\ref{eq:mod1 balance_ck}) using
Equation~(\ref{eq:mod1 balance_ck-1}). After imposing normalization, we 
obtain

\begin{equation}
P(\vec{c})=P(c_{1},\dots,c_{k},\dots)=\frac{1}{Z_{\alpha,r}}
\left(\frac{\alpha}{r}\right)^{C}\frac{r^{N}}{\prod_{n=1}^{N}\mu(n)}
\frac{1}{\prod_{i=1}^{\infty}i^{c_{i}}c_{i}!},
\label{PvecCC}
\end{equation}
where $C=\sum_{k\geq1}c_{k}$ as defined in Equation~(\ref{eq:relations NCc})
and where $Z_{\alpha,r}$ is the normalization constant that can be
obtained by evaluating $P(N=0)$ in Equation~(\ref{eq:mod3 PN}) so that

\begin{equation}
Z_{\alpha,r}=\sum_{c_{1},\dots,c_{k},\dots}^{\infty}\left(\frac{\alpha}{r}\right)^{C}
\frac{r^{N}}{\prod_{n=1}^{N}\mu(n)}\frac{1}{\prod_{i=1}^{\infty}i^{c_{i}}c_{i}!}
=1+\sum_{n=1}^{\infty}\frac{1}{n!}{\displaystyle 
\prod_{k=0}^{n-1}}\frac{\alpha+rk}{\mu(k+1)}.\label{Znorm}
\end{equation}
More details can be found in Appendix~\ref{subsec:S-mod3-Distrib N}.
We can now use the expression for $P(\vec{c})$ in
Equation~(\ref{PvecCC}) to evaluate the moment generating function of
$C$ and related moments
\[
M_{C}(\xi)\equiv\E{\exp\left(\xi C\right)}=\frac{1}{Z_{\alpha,r}}
\sum_{c_{1},\dots,c_{k},\ldots}\left(\frac{\alpha}{r}e^{\xi}\right)^{C}
\frac{r^{N}}{\prod_{n=1}^{N}\mu(n)}\frac{1}{\prod_{i=1}^{\infty}i^{c_{i}}c_{i}!}.
\]
Since the argument of the sum in the above expression is the same
as in Equation~(\ref{Znorm}) provided $\alpha\to\alpha e^{\xi}$
we can write

\begin{align*}
M_{C}(\xi) & =\frac{Z_{\alpha e^{\xi},r}}{Z_{\alpha,r}},
\end{align*}
for any $\xi<0$. We can now differentiate $M_{C}(\xi)$ with respect
to $\xi$ and take the limit $\xi\rightarrow0$ to find the
following expressions for the mean and the variance of $C$
\begin{align}
\E C= & \alpha\,\E{\sum_{k=0}^{N-1}\left(\alpha+rk\right)^{-1}}, \label{eq:mod3 EC}\\
\var C= & \E C\left(1-\E C\right)+\alpha^{2}\,\E{\left(\sum_{k=0}^{N-1}
\frac{1}{\alpha+r k}\right)^{2}-\sum_{k=0}^{N-1}
\frac{1}{\left(\alpha+r k\right)^{2}}}.
\label{eq:mod3 VC}
\end{align}

\noindent We can use the above expressions and $P(N)$ as determined in
Equation~(\ref{eq:mod3 PN}) to evaluate the mean and variance of
$C$. Note that setting a uniform $\mu(N)=\mu$ in
\added{Equations}~(\ref{PvecCC}) \added{and}~(\ref{Znorm}) reduces the results to
those of the sBDI model (Section~\ref{sec:BDIanalyze}).  We can now
evaluate $\E{c_{k}}$ using Equation~(\ref{PvecCC}):
\[
\E{c_{k}}=\frac{1}{Z_{\alpha,r}}\!\sum_{c_{1},\dots,c_{k},\dots}\!\!c_{k}
\left(\frac{\alpha}{r}\right)^{C}\frac{r^{N}}{\prod_{n=1}^{N}\mu(n)}
\frac{1}{\prod_{i=1}^{\infty}i^{c_{i}}c_{i}!},
\]
which can be rearranged to yield
\begin{align}
\E{c_{k}} & =\frac{1}{Z_{\alpha,r}}\frac{\alpha}{kr}\!\sum_{c_{1},\dots,c_{k},\dots}
\!\!\left(\frac{\alpha}{r}\right)^{C}\frac{r^{N+k}}{\prod_{n=1}^{N+k}\mu(n)} 
\frac{1}{\prod_{i=1}^{\infty}i^{c_{i}}c_{i}!}\nonumber \\
 & =\frac{\alpha r^{k-1}}{k}\sum_{c_{1},\dots,c_{k},\dots}\!\!
\frac{P(\vec{c})}{\prod_{m=1}^{k}\mu(N+m)}\nonumber \\
 & =\frac{\alpha r^{k-1}}{k}\E{\prod_{m=1}^{k}
\frac{1}{\mu(N+m)}}.\label{eq:mod3 Eck}
\end{align}
A uniform $\mu(n)=\mu$ returns $\E{c_{k}}=(\alpha/\mu) (r/\mu)^{k-1}/k$,
as previously determined in Section~\ref{sec:BDIanalyze}. We can
also verify that for any function $f(N)$,

\[
\E{c_{k}f(N)}=\frac{\alpha r^{k-1}}{k}\E{\frac{f(N+k)}{\prod_{m=1}^{k}\mu(N+m)}}.
\]
For $f(x)=\log\left(x/k\right)/x$
and $f(x)=(k/N)^{2}$ the expressions for Shannon's
Entropy and Simpson's diversity index become 

\begin{eqnarray}
\E H & = & \alpha\sum_{k=1}^{\infty}r^{k-1}
\E{\frac{\log\left[\frac{N+k}{k}\right]}{\left(N+k\right)\prod_{m=1}^{k}\mu(N+m)}},\\
\E S & = & 1-\alpha\sum_{k=1}^{\infty}kr^{k-1}
\E{\frac{1}{\left(N+k\right)^{2}\prod_{m=1}^{k}\mu(N+m)}}.
\label{eq:mod3 EH ES-1}
\end{eqnarray}

\noindent Once again, setting $\mu(N)=\mu$ a constant allows us to
recover the results in Equation~(\ref{eq:mod1 EH ES}) for the sBDI
model.

\subsection{Fast immigration limit}

%
%
%

To analyze the large immigration limit, $\alpha=\widetilde{\alpha}\,\Omega$, 
$\Omega\to\infty$, we need to assume a specific form for the death rate. 
For a given $\Omega$, we take the death rate as a
function of $N/\Omega$:
\[
\mu(N)=\widetilde{\mu}(N/\Omega).
\]
The reason behind this scaling is that we want to keep $\mu(N)$ at the
same order of magnitude as $\Omega$ increases. As in the previous
models, we will show that $\E{N}$ diverges as $\Omega$ increases, but
the random variable $N/\Omega$ will be shown to converge in
distribution to a constant. As a consequence, the death rate
$\widetilde{\mu}(N/\Omega)$ will also converge in distribution to a
constant.

Given $\widetilde{\mu}(x)$ is continuous and strictly increasing, and
that  $\lim_{x\rightarrow\infty} r/{\widetilde{\mu}}(x) < 1$,
one can show that there exists a unique, positive solution
$n^{*}$ to the fixed-point equation

\begin{equation}
n^{*}\widetilde{\mu}(n^{*})=\widetilde{\alpha}+rn^{*}.
\label{eq:fixed-pt-eq}
\end{equation}
In Appendix~\ref{subsec:S-mod3-Convergence N_Omega}, we show that
for every $\delta >0$,
\[
P(|N/\Omega-n^{*} |>\delta)\xrightarrow{\Omega\rightarrow\infty}0,
\]
thus proving that
\begin{equation}
\frac{N}{\Omega}\xrightarrow[\Omega\rightarrow\infty]{{\cal D}}n^{*}
\label{eq:mod3 conv N/=0003A9}
\end{equation}
in which $n^{*}$ is defined by Equation~(\ref{eq:fixed-pt-eq}). The
proof of this convergence is analogous to the one in \citet[Proposition 4]{Dessalles2017b}).
Intuitively, $n^{*}$ can be identified with the steady-state solution
to the deterministic approximation of the dynamics of $n(t) \equiv N (t)/\Omega$
given by
\[
\frac{{\rm d} n(t)}{{\rm d}t}=\widetilde{\alpha}+r n(t)-\widetilde{\mu}(n(t)) n(t).
\]

\noindent
Using the convergence of Equation~(\ref{eq:mod3 conv N/=0003A9}),
we can show convergence in distribution of $C/\Omega$ and $c_{k}/\Omega$
as follows
\[
\frac{C}{\Omega}\xrightarrow[\Omega\rightarrow\infty]{{\cal D}}
\frac{\widetilde{\alpha}}{r}\log
\left[\frac{1}{1-r/\widetilde{\mu}(n^{*})}\right]\quad\text{and}\quad
\frac{c_{k}}{\Omega}\xrightarrow[\Omega\rightarrow\infty]{{\cal D}}
\frac{\widetilde{\alpha}}{r}\,\frac{1}{k}\left(\frac{r}{\widetilde{\mu}(n^{*})}\right)^{k}.
\]
The complete proofs of these convergences are given in
Appendix~\ref{subsec:S_mod3 Convergence C/=0003A9}, but one can also
verify them by inspecting Equations~(\ref{eq:mod3 EC})
and~(\ref{eq:mod3 Eck}) respectively to determine the convergence of
$\E{C/\Omega}$ and $\E{c_{k}/\Omega}$ using convergence of $N/\Omega$
to $n^{*}$.

Even though the dynamics of all $n_{i}$ are coupled through the death
rate $\mu(N)=\mu\left(\sum n_{i}\right)$, all $n_{i}$ remain
identically distributed: $P(n_{i}=k)=P(n_{j}=k)$ for all $i,j\leq C$
and $k\geq1$. This ``neutrality'' allows us to determine the convergence of
$n_{i}$ in the $\Omega \to \infty$ limit as detailed in
Appendix~\ref{subsec:S_mod3 Convergence ni}:

\[
P(n_{i}=k)\xrightarrow[\Omega\rightarrow\infty]{}
\frac{1}{k}\left(\frac{r}{\widetilde{\mu}(n^{*})}\right)^{k}
\frac{-1}{\log[1-r/\widetilde{\mu}(n^{*})]},
%
\]
which shows that for $\Omega \to \infty$, $n_{i}$ converges to a
logarithmic-series distribution with parameter $r/\mu(n^{*})$.

Finally, we can use the convergence in distribution of both $N/\Omega$
and $c_{k}/\Omega$, to determine the convergence in distribution for
the rescaled Shannon's Entropy $H/\log\Omega$ and Simpson's diversity
index $S$:
\[
\frac{H}{\log\Omega}\xrightarrow[\Omega\rightarrow\infty]{{\cal D}}1
\quad\text{and}\quad S\xrightarrow[\Omega\rightarrow\infty]{{\cal D}}1.
\]
These convergence results are identical for all three models and their
proofs are similar to the ones for the sBDI model as described in
Appendix~\ref{subsec:S-mod1-Convergence high immigration}.

\subsection{Interpretation and analysis of results}


\begin{figure}
\begin{center}\includegraphics[width=6.8in]{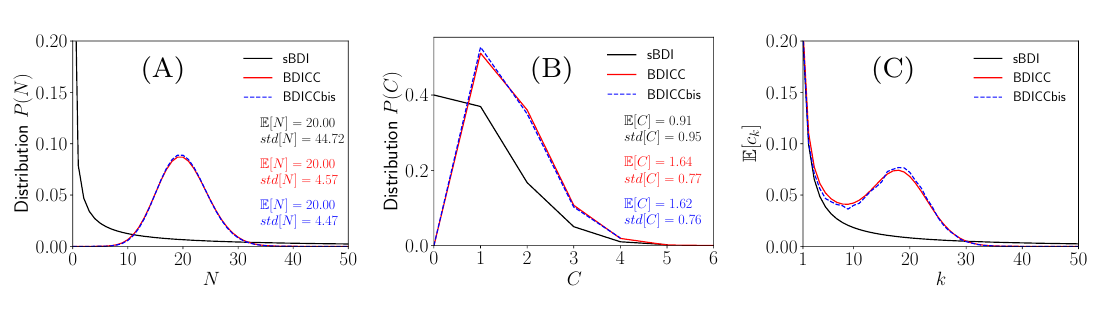}\vspace{-5mm}
\caption{\label{fig:mod3 BDICC vs sBDI}Comparison of the sBDI model
  with the two carrying capacity models (BDICC with carrying capacity
  on the death rate and BDICC-bis with carrying capacity on birth the
  rate) for slow immigration (parameters are chosen such as
  $\protect\E N=20$ in both cases: $\alpha=0.2$, $r=.99$, $\mu=1$ for
  the sBDI; $\alpha=0.2$, $r=.99$, $\mu(N)=0.0475 N$ for the
  BDICC). (a) Theoretical distributions $P(N)$ for the three
  models. (b) Distributions of the richness $P(C)$ obtained from
  Monte-Carlo simulations. (c) The theoretical expected species-count
  vector $\left(\protect\E{c_{k}}\right)_{k\protect\geq1}$ calculated
  from Equations~(\ref{eq:mod3 Eck}) and (\ref{eq:mod3 PN}).  Contrary
  to the sBDI model, the BDICC model is dominated by only one species
  ($C\simeq1$) with around $20$ individuals (the peak of
  $\protect\E{c_{k}}$ arises at $k\simeq20$). This attribute is completely
  missed in a mean-field approximation to $\E{c_{k}}$
  \citep{GOYALBMC}. Negligible differences between the BDICC and
  BDICC-bis models are observed.}
\end{center}
\end{figure}

\subsubsection{Comparison with the sBDI model}
\label{subsec:mod3 CC on death}

To properly compare the sBDI and BDICC models, we fix their
immigration rates $\alpha$ and birth rates $r$ to be the same. For the
BDICC model, we use a linear death rate function $\mu(N)=\mu_{1} N$
and tune both $\mu_{1}$ and the constant death rate $\mu$ in the sBDI
model to yield the same average total number of individuals $\E N$.

In Figure~\ref{fig:mod3 BDICC vs sBDI} we plot the distributions
$P(N)$ and $P(C)$ as well as the average $\E{c_{k}}$ in a low
immigration regime ($\alpha=0.2$ and $r=0.99$) for both the sBDI and
the BDICC models. We adjusted $\mu$ for the sBDI model and $\mu_{1}$
for the BDICC model so that $\E{N}=20$ in both cases. Clear
differences emerge. First, since $\mu(N)$ is proportional to the
existing population $N$ in the BDICC model, very rarely will the
population reach vanishingly small levels: as $N \to 0$ so will
$\mu(N)\to 0$ allowing birth and immigration to replenish $N$. This is in
contrast to the sBDI model where $\mu$ is a constant independent of
$N$.

Another feature of a low immigration rate is that it allows one
species to ``invade the niche'' of the BDICC model before the arrival
of another species.  The result is that only one species ($C\simeq1$)
represents the whole population and $\E{c_{k}}$ has a peak around $k
\approx \E{N}=20$.  This exclusion effect does not arise in the sBDI
model since the presence of species already in the system does not
influence the dynamics of the newly arriving ones. These exclusionary
interactions are also \added{the} origin of the peak observed in
Figure~\ref{fig:mod3 BDICC vs sBDI}(c).  \added{Note that this
  difference is not only due to the sBDI model's high probability of
  extinction ($N=0$): we checked that the distributions of the sBDI
  model, conditioned on $N>0$, also fail to display the exclusionary
  effect where one clone dominates.} Direct mean-field approximations,
$\E{c_{k}c_{\ell}}\approx \E{c_{k}}\E{c_{\ell}}$, lead to monotonic
decreasing $\E{c_{k}}$ \citep{GOYALBMC}, completely missing the peak
around the carrying capacity ($k \approx 20$). Global carrying
capacity interactions can also have a significant influence on
Shannon's entropy and Simpson's diversity index.


The qualitative differences between the two models diminish as the
immigration rate $\alpha$ increases. This confirms our theoretical
analysis through which we showed that the sBDI and the BDICC models
follow similar trends as $\alpha$ increases.
%
%
If we fix $\mu$ of the sBDI model and $\mu_{1}$ in the BDICC model
such that $\lim_{\Omega \to \infty} \E{N/\Omega}$ remains the same for
both models, we find that $N/\Omega$, $C/\Omega$ and $c_{k}/\Omega$
converge to the same constants in the two models and that $n_{i}$
converges to the same the log-series distribution as well.

\subsubsection{\label{subsec:mod3 CC on birth}Carrying capacity on birth (BDICC-bis
model)} Our BDICC model included an interaction only through the death
rate $\mu(N)$. This choice, as opposed to, say, $r(N)$ was made
because the detailed balanced assumption can be shown to 
hold between all pairs of states, rendering 
our analytic results for the probability distribution $P(\vec{c})$
exact.

Alternatively, one can impose an interaction through a
population-dependent birth term. It is well-known that even if the
mean populations are equal, models using $\mu(N)$ yield different
higher order statistics from those using $r(N)$ \citep{LJSALLEN}.  The
interacting model with $\mu$ constant, but a growth rate $r(N)$ is
dubbed the BDICC-bis model. For the BDICC-bis model, the equilibrium
distribution of $N$ can still be determined as
\[
P(N) = \begin{cases}
\displaystyle {1\over Z_{\alpha,\mu}},\quad N=0, \\
\displaystyle \frac{1}{Z_{\alpha,\mu}}\frac{1}{N!}
\prod_{k=0}^{N-1}\frac{\alpha+r(k) k}{\mu}, \quad N\geq 1,
\end{cases}
\]
with $Z_{\alpha,\mu}$ a normalizing constant. However, as shown in
Section~\ref{subsec:S_mod3 BDICC on birth} of the Appendix the
BDICC-bis model with population-dependent growth is no longer
reversible when enumerated by the species counts $c_{k}$ and we cannot
use detailed balance properties to exactly determine the probability
distribution $P(\vec{c})$. Consequently, neither \deleted{can} means
\added{nor} variances of $c_{k}$ and $C$ \added{can} be
determined. We thus perform numerical simulations by setting
$r(N)=r_{1}/N$, while keeping $\alpha, \mu$ uniform.

We compare results of the BDICC-bis model to those of the sBDI model
($\alpha, r, \mu$ uniform) and the previous BDICC model ($\alpha, r,
\mu(N) = \mu_{1} N$). As in Subsection~\ref{subsec:mod3 CC on death}
we consider a low immigration rate $\alpha=0.2$, set $\mu =1$, and
adjust the parameter $r_{1}$ so that $\E{N}$ is the same across the
three models. Results for the BDICC-bis model are plotted as the
\added{blue} dashed curves in Figure~\ref{fig:mod3
  BDICC vs sBDI}.  Observed trends for the $P(N)$ and $P(C)$
distributions within the BDICC and the BDICC-bis models are similar,
as well as for $\E{c_{k}}$.  Shannon's entropy and Simpson's diversity
index also remain similar, $\E H=0.25$ and $\E S=0.15$ for the
BDICC-bis model, and $\E H=0.26$ and $\E S=0.16$ for the BDICC model.

\subsubsection{\label{reflectingBC} Quasi-steady state and reflecting boundary conditions} 
When $\alpha = 0$ in the BDICC model, the $N=0$ state is a perfect
sink. In the absence of immigration, a system cannot escape from the
``absorbing'' $N=0$ state. However, in the deterministic limit, the
$N=0$ state is unstable while the finite-population state with $N^{*}$
individuals is stable (for $\mu(N) = \mu_{1}N$, $N^{*} = r/\mu_{1}$).
Even though the true steady-state of the stochastic problem is $N=0$,
it may take an exponentially long time for a population initially at
$N\sim N^{*}$ to become extinct. Therefore, given a system initiated
with a large population $N \sim N^{*}$, we expect that a quasi-steady
state is established before extinction.

To find distributions associated with the long-lived quasi-steady
state of the BDICC model, we modify the absorbing boundary condition
at $N=0$ to a reflecting boundary condition by simply preventing the
last individual from dying by setting $\mu(N=1) = 0$. 
%
%
We can now compute the steady state distribution of $N$ using detailed
balance to find
\[
P(N)=P(1)\,\frac{1}{N!}\prod_{k=1}^{N-1}\frac{\alpha+rk}{\mu(k+1)},
\]
with $P(1)$ being the probability of having one individual.  Contrary
to the BDICC model with an absorbing boundary condition, we can no
longer recurse the detailed balance equations down to $N=0$, since the
last individual cannot die (in other words $P(0)=0$). By denoting
\[
Z_{\alpha,r}'=\sum_{n=1}^{\infty}\frac{1}{n!}\prod_{k=1}^{n-1}\frac{\alpha+rk}{\mu(k+1)},
\]
we find
\[
P(N)=\frac{1}{Z_{\alpha,r}'}\,\frac{1}{N!}\prod_{k=1}^{N-1}\frac{\alpha+rk}{\mu(k+1)}.
\]
Similarly, using the 
detailed balance equations, we find the distribution of $\vec{c}$,

\begin{align*}
P(\vec{c}) & =\frac{P(1,0,\ldots)}{\alpha}
\left(\frac{\alpha}{r}\right)^{C}\frac{r^{N}}{\prod_{n=2}^{N}\mu(n)}
\frac{1}{\prod_{i=1}^{\infty}i^{c_{i}}c_{i}!}, 
\end{align*}
where $P(1,0,\ldots)=1/Z'_{\alpha,r}$.


The importance of the quasi-steady state is most discernible in the
$\alpha \to 0$ limit where initial conditions determine long-lived
configurations. With absorbing boundary conditions, the equilibrium
state is the trivial empty state even if it is deterministically
unstable.  However, \added{by using a reflecting boundary condition on
  the total population, we can approximate the long-lived quasi-steady
  state distributions with}

\begin{align}
P(N) & \xrightarrow[\alpha\to0]{}\frac{1}{Z_{0,r}'}\,
\frac{r^{N-1}}{N!}\prod_{k=2}^{N}\frac{1}{\mu(k+1)} \label{REFLECTINGPN} \\
P(\vec{c}) & \xrightarrow[\alpha\to 0]{}
\begin{cases}
\displaystyle \frac{1}{Z_{0,r}'}\,\frac{r^{N-1}}{N!}\prod_{k=2}^{N}
\frac{1}{\mu(k+1)} & \text{if}\quad C=\sum_{k}c_{k}=1 \label{REFLECTINGPc}\\
0 & \text{otherwise}.
\end{cases}
\end{align}
In this limit, only one species survives and occupies the whole
system before final extinction at exponentially long times. 

Intuitively, without immigration, new species cannot be introduced in
the system, and with probability $1$ there will be at some point only
one individual in the system. This single-species population persists
for a long time before final extinction.  This long time persistence
is approximated by the reflecting boundary condition that prevents
true extinction.  Note that this limit is related to species
extinction and coarsening in a multispecies Moran model with fixed
population size \citep{BLYTHE}. The distributions $P(N)$ for absorbing
and reflecting boundary conditions are compared in
Figure~\ref{fig:QSS} for small $\alpha$.

\begin{figure}[h!]
\begin{minipage}[l]{0.46\linewidth}
\includegraphics[height=5cm]{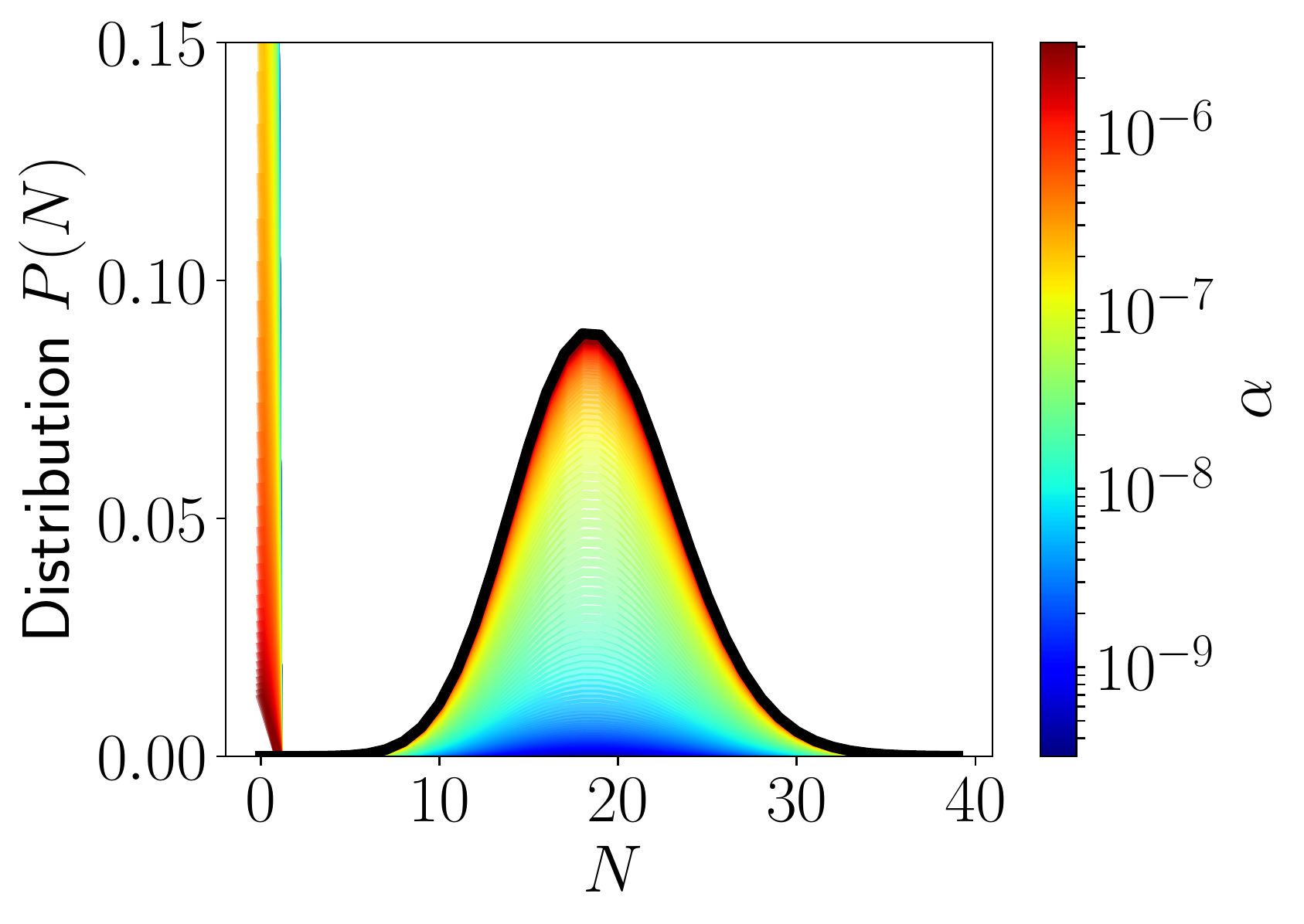}
\end{minipage}
\begin{minipage}[l]{0.53\linewidth}
\caption{\baselineskip=11pt Comparison of $P(N)$ for the BDICC model
  with absorbing and reflecting boundary conditions in the small
  $\alpha$ limit. For both submodels, $r=0.995$ and $\mu(N) = 0.0498
  N$, leading to a carrying capacity of $N^{*}\approx 20$.  The thick
  black curve corresponds to the quasi-steady state for $\alpha = 0$
  computed by using a reflecting boundary condition approximation
  (Eq.~\ref{REFLECTINGPN}). The colored curves correspond to the
  steady-state distribution of the absorbing model using different
  values of $\alpha$. When $\alpha = 0$, the standard absorbing BDICC
  model leads to an equilibrium ``vacuum'' or ``extinct'' state (dark
  blue), while the BDICC model approximated with reflecting boundary
  condition leads to a the quasi-steady state distribution $P(N)$
  centered about $N^{*}$.}
 \label{fig:QSS}
\end{minipage}
\end{figure}

\begin{doublespace}
{\footnotesize{}}
\begin{table}
\begin{doublespace}
\begin{centering}
\resizebox{!}{0.9\textheight}{
\begin{turn}{90}{\footnotesize{}}%
\begin{tabular}{|c|>{\centering}m{5cm}|>{\centering}m{5.75cm}|>{\centering}m{6.25cm}|}
\hline & Simple Birth-Death-Immigration model (sBDI) &
Birth-Death-Immigration model with mutation (BDIM) &
Birth-Death-Immigration model with carrying capacity
(BDICC)
\tabularnewline 
\hline 
\hline Defn. & \begin{onehalfspace}
  $u\equiv \alpha/r$, $v\equiv r/\mu$
\end{onehalfspace}

$Z=\left(1-v\right)^{-u}$

$f_{k}(x)=\log\left[\frac{x+k}{k}\right]/\left(x+k\right)$ & $u=\alpha/r$, $v=r/\mu$

$p\equiv \frac{r}{\mu}\left(1-\epsilon\right)$ & 
$u\equiv \alpha/r$, $v(x) \equiv r/\mu(x)$

$Z_{\alpha, r}=\sum_{n=0}^{\infty}\frac{1}{n!}{\displaystyle \prod_{k=0}^{n-1}}
\frac{\alpha+r k}{\mu(k+1)},$

$S_{n}(x)=\sum_{k=0}^{x-1}(\alpha+r k)^{-n}$

$f_{k}(x)=\log\left[\frac{x+k}{k}\right]/\left(x+k\right)$\tabularnewline
\hline
Cond. & $v<1$ & $v<1$ & $\lim_{x\rightarrow\infty}v(x)<1$~(sufficient cond.{*})\tabularnewline
\hline
$N$ & $N\sim\text{NegBinom}(u,v)$ & $N\sim\text{NegBinom}(u,v)$ & 
$P(N)=\frac{1}{Z_{\alpha,r}}\frac{1}{N!}{\displaystyle \prod_{k=0}^{N-1}}
\left(u+k\right) v(k+1)$\tabularnewline
\hline
$\vec{c}$ & $P(\vec{c})=\frac{1}{Z} \frac{u^{C}v^{N}}
{\prod_{i=1}^{\infty}i^{c_{i}}c_{i}!}$ & N.A. & 
$P(\vec{c})=\frac{u^{C}}{Z_{\alpha,r}}{\displaystyle 
\prod_{n=1}^{N}}\frac{v(n)}{\prod_{i=1}^{\infty}i^{c_{i}}c_{i}!}$\tabularnewline
\hline
$c_{k}$ & $c_{k}\sim\text{Poisson}\left(\frac{u v^{k}}{k}\right)$ & 
$\E{c_{k}}=\frac{uv}{1-v}\left(1-p\right)\frac{p^{k-1}}{k}$

$\var{c_{k}}=\E{c_{k}}+\epsilon \frac{uv^{2} p^{2(k-1)}}{(1-v)^2 k^{2}}$ & 
$\E{c_{k}}=\frac{u}{k}\E{{\displaystyle \prod_{m=1}^{k}}v(N+m)}$\tabularnewline
\hline
$C$ & $C\sim\text{Poisson}\left(u\log\left[\frac{1}{1-v}\right]\right)$ & 
$\E C=\frac{uv}{1-v}\frac{1-p}{p}\log\frac{1}{1-p}$
\vspace{2mm}
$\begin{alignedat}{1}\vspace{2mm}\var C= & \E C\left[\frac{\E
      C}{u}+1+\log(1-p)\right]\end{alignedat} $ & $\E
C=\alpha\E{S_{1}(N)}$

$\begin{alignedat}{1}\var C= & \E C\left(1-\E C\right)\\[-8pt] &
 +\alpha^{2}\E{\left(S_{1}(N)\right)^{2}-S_{2}(N)}
\end{alignedat}
$\tabularnewline
\hline
$\vec{n}$ & $n_{i}\sim\text{LogSeries}\left(v\right)$ & $n_{i}\sim
\text{LogSeries}\left(p\right)$ & N.A.\tabularnewline
\hline
$H$ & $\E H=u{\displaystyle \sum_{k=1}^{\infty}}v^{k}\E{f_{k}(N)}$ & $\E H
\simeq\left(1-p\right){\displaystyle \sum_{k=1}^{\infty}}p^{k-1}\log
\left[\frac{uv}{k\left(1-v\right)}\right]$ & $\E H=
u{\displaystyle \sum_{k=1}^{\infty}}\E{f_{k}(N){\displaystyle 
\prod_{m=1}^{k}}v(N+m)}$\tabularnewline
\hline
$S$ & $\E S=1-u{\displaystyle \sum_{k=1}^{\infty}}kv^{k}
\E{\left(\frac{1}{N+k}\right)^{2}}$ & $\E S\simeq1-\frac{1}{uv}\frac{1-v}{1-p}$ & 
$\E S=1-u{\displaystyle \sum_{k=1}^{\infty}}k\E{\frac{\prod_{m=1}^{k}v(N+m)}
{\left(N+k\right)^{2}}}$\tabularnewline
\hline
\end{tabular}\end{turn}}
\par\end{centering}
\end{doublespace}
{\footnotesize{}\caption{\label{tab:recap}Table summarizing our
    analytical results. $\text{Poisson}$: Poisson distribution;
    $\text{NegBinom}$: Negative binomial distribution,
    $\text{LogSeries}$: logarithmic distribution. In each case, the
    quantities $C$ and $N$ implicitly depend on the vector $\vec{c}$
    through Equation~(\ref{eq:relations NCc}). ({*} indicates a
    sufficient condition, for a necessary and sufficient condition,
    see \citet[Chapter 1]{Bansaye2015a}.) The functions
      $f_{k}(x)$, $S_{1}(N)= \sum_{k=0}^{N-1}(\alpha + rk)^{-1}$ and
      $S_{2}(N)=\sum_{k=0}^{N-1}(\alpha + rk)^{-2}$ are defined in
      entries of the first row.}}
\end{table}
\end{doublespace}

{\footnotesize{}}
\begin{table}
\begin{doublespace}
\begin{centering}
\resizebox{!}{0.9\textheight}{
\begin{turn}{90}{\footnotesize{}}%
\begin{tabular}{|c|>{\centering}m{5cm}|>{\centering}m{5.75cm}|>{\centering}m{6.25cm}|}
\hline \:  & Simple Birth-Death-Immigration model (sBDI) &
Birth-Death-Immigration model with mutation (BDIM) &
Birth-Death-Immigration model with carrying capacity
(BDICC)\tabularnewline \hline \hline Cond. &
$\Omega\rightarrow\infty$ & $\Omega\rightarrow\infty$ &
$\Omega\rightarrow\infty$

$\mu(x)=\widetilde{\mu}(x/\Omega)$

$n^{*}$: positive soln of
$x\widetilde{\mu}(x)=\widetilde{\alpha}+r\,x$\tabularnewline
\hline $N/\Omega$ & $\frac{\widetilde{\alpha}}{\mu-r}$ &
$\frac{\widetilde{\alpha}}{\mu-r}$ & $n^{*}$\tabularnewline \hline
$c_{k}/\Omega$ & $\frac{\widetilde{\alpha}}{r}\,\frac{(r/\mu)^{k}}{k}$
&
$\frac{\widetilde{\alpha}}{\mu-r}\,\frac{\mu-r\left(1-\epsilon\right)}
{r\left(1-\epsilon\right)}\,\frac{\left(\frac{r}{\mu}\left(1-\epsilon\right)\right)^{k}}{k}$
&
$\frac{\widetilde{\alpha}}{k}\,\frac{r^{k-1}}{\widetilde{\mu}(n^{*})^{k}}$\tabularnewline
\hline $C/\Omega$ &
$\frac{\widetilde{\alpha}}{r}\log\left[\frac{1}{1-r/\mu}\right]$ &
$\frac{\widetilde{\alpha}}{\mu-r}\frac{\mu-r\left(1-\epsilon\right)}
{r\left(1-\epsilon\right)}\log\left[\frac{1}{1-\frac{r}{\mu}\left(1-\epsilon\right)}\right]$
&
$\frac{\widetilde{\alpha}}{r}\log\left[\frac{1}
{1-r/\widetilde{\mu}(n^{*})}\right]$\tabularnewline
\hline $n_{i}$ & $n_{i}\sim\text{LogSeries}\left(r/\mu\right)$ &
$n_{i}\sim\text{LogSeries}\left(\frac{r}{\mu}\left(1-\epsilon\right)\right)$
&
$n_{i}\sim\text{LogSeries}\left(r/\widetilde{\mu}(n^{*})\right)$\tabularnewline
\hline $H/\log\Omega$ & $1$ & $1$ & $1$\tabularnewline 
\hline 
$S$ & $1$ & $1$ & $1$
\tabularnewline \hline
\end{tabular}\end{turn}}
\par\end{centering}
\end{doublespace}
{\footnotesize{}\caption{\label{tab:recap-1}Table summarizing model results in 
the fast immigration limit defined by 
$\alpha = \widetilde{\alpha}\Omega, \Omega \to \infty$. $H/\log\Omega$ and $S$
are expanded to the first nontrivial term.}}{\footnotesize\par}
\end{table}
{\footnotesize\par}


\section{Summary and Conclusions}

In this paper we analyzed three stochastic, neutral
birth-death-immigration (BDI) models: the simple BDI (sBDI), BDI with
mutations (BDIM), and BDI with carrying capacity (BDICC). Where
possible, we derived analytical expressions for the steady-state
distribution $P(N)$ of the total population and the steady-state
distribution $P(C)$ for the total number of species in the system. In
many cases, we were also able to derive expressions for the
steady-state distributions of individual subpopulations $P(n_{i})$ and
$P(c_{k})$, given in terms of cells counts $n_{i}$ and species counts
$c_{k}$, respectively.

All three models (sBDI, BDIM, and BDICC) analyzed show similar species
abundance distribution functions. In particular, we find that the
number of individuals in one species $n_{i}$ follows a strict
log-series distribution $P(n_i)$, or, in the case of the BDICC model,
can be approximated by one.  The prediction that species could follow
this type of distribution dates to the early days of theoretical
ecology. For example, after analyzing insect abundances in the field,
\citet{Fisher1943} proposed that the distribution of insect species in
an area should follow a geometric or, possibly, a log-series
distribution.  The log-series distribution has since been widely used
in theoretical ecology \citep{Volkov2003,Bell2001,MacArthur2016}, but
has also been challenged. For instance \citet{Preston1948} speculated
that actual species abundances would be better described by a
log-normal, or possibly a Poisson log-normal distribution
\citep{Bulmer1974}. Within immunology, the abundance of T-cell clones
appears to follow a power-law distribution, incompatible with a
log-series distributions \citep{Desponds2016}. The log-series
characteristic of our BDI models can be linked to their neutrality,
{\textit {i.e.}} that replication and death rates are independent of
the given species.

We also evaluated diversity metrics such as Shannon's entropy and
Simpson's diversity index and provided expectations and variances of a
number of quantities. Stochastic simulations were also performed and
matched with our analytical results. Our analytical results are
summarized in Table~1, while Table~2 lists the same results in the
large immigration regime. Interestingly, we show that in the fast
immigration limit, the diversity indices $H$ and $S$ converge to
values independent of the model, but the richness $C$ converges to
values that are model-dependent. Only the richness can distinguish the
different processes in the fast immigration limit, implying that 
in this limit it is a more useful diversity metric.

Finally, we confirmed the consistency of detailed balance for a
carrying capacity model in which the global interaction is implemented
through the death rate (BDICC) but demonstrated that detailed balance
is violated if carrying capacity is effected through the birth rate
(BDICC-bis model).  Nonetheless, this asymmetry generates almost no
qualitative difference in the statistical properties when comparing
the two models using equal mean total populations.

Many related applications motivate us to extend our work towards
non-neutral BDI models. We expect that lifting the neutrality
condition will typically generate longer tails in species abundance
distributions.

\section*{Acknowledgments}
This work was supported in part by an INRA Contrat Jeune Scientifique
Award (RD) and by the National Science Foundation through grants
DMS-1814364 (TC) and DMS-1814090 (MD).  The authors also thank Song Xu
for clarifying discussions.



\appendix

\section*{Mathematical Appendices}

\section{Simple Birth-Death-Immigration models (sBDI)}

\subsection{\label{FINITENUMBER} Finite number of
    species} So far, we have assumed immigration events introduce
  completely new species to the system, regardless of the existing
  population structure. Within the context of island biodiversity,
  this assumption corresponds to the mainland hosting an unlimited
  number of species, so that individuals who emigrate to the island
  are always part of a new species. Mathematically, we are assuming
  that each species immigrates only once.

In this Appendix, we consider an alternative model where the number of
mainland species $Q$ is finite.  In this case, the probability that a
newly immigrated individual belongs to species $i$ (with $1\leq i\leq
Q$) is $1/Q$ and the number of species in the island cannot exceed
$Q$. As a consequence, the total number of species $C \leq Q$, and the
number of species with $k$ individuals $c_{k}\leq Q$ for all $k$.

The dynamics of the total number of individuals $N$ remains unchanged
with respect to the sBDI model, as the type of species immigrating
from the mainland does not affect overall birth or death
rates. Therefore, the distribution for $P(N)$ remains identical to the
one derived in Eq.~(\ref{eq:mod1 PN}) for the simple BDI model.  We
can now determine the distribution of $\vec{c}$ in the alternative
model using the same approach taken for the sBDI model. Transitions
are given by

{\small{}
\begin{align*}
 &  & (c_{1},c_{2},\ldots) & \xrightarrow{\alpha(1-C/Q)}(c_{1}+1,c_{2},\ldots) & 
\left.\vphantom{\begin{array}{c}
(c_{1})\\
(c_{1})
\end{array}}\right\}  & \begin{alignedat}{1} & \text{Immigration}\\
 & \text{(of new species)}
\end{alignedat}
\\
\text{for }k\geq1\quad &  & (c_{1},\ldots,c_{k},c_{k+1},\ldots) & 
\xrightarrow{\left(rk+\alpha/Q\right)c_{k}}(c_{1},\ldots,c_{k}-1,c_{k+1}+1,\ldots) & 
\left.\vphantom{\begin{array}{c}
(c_{1})\\
(c_{1})
\end{array}}\right\}  & \begin{alignedat}{1} & \text{Birth+\text{Immigration }}\\
 & \text{(of existing species)}
\end{alignedat}
\\
\begin{aligned}\text{for }k\geq2\quad\\
\\
\end{aligned}
 &  & \begin{aligned}(c_{1},\ldots,c_{k-1},c_{k},\ldots)\\
(c_{1},c_{2},\ldots)
\end{aligned}
 & \begin{aligned} & \xrightarrow{\mu kc_{k}}(c_{1},\ldots,c_{k-1}+1,c_{k}-1,\ldots)\\
 & \xrightarrow{\mu c_{1}}(c_{1}-1,c_{2},\ldots)
\end{aligned}
 & \left.\vphantom{\begin{array}{c}
(c_{1})\\
(c_{1})
\end{array}}\right\}  & \,\,\text{Death}.
\end{align*}

\noindent
Note that the birth process rate is effectively augmented by
$\alpha/Q$, due to the possibility of a new individual immigrating
into an existing species. Conversely, the corresponding immigration
rate for new species is decreased by $\alpha C/Q$. Also note that the
limit $Q \to \infty$ reduces the current model to the original
sBDI. Using detailed balanced equations, similarly as in the sBDI
model, we can write $P(\vec{c})$ as follows
\[
P(\vec{c})=\left(1-\frac{r}{\mu}\right)^{\alpha/r}\frac{Q!}{\left(Q-C\right)!}
\left(\frac{r}{\mu}\right)^{N}
\left(\frac{1}{\prod_{i=1}^{\infty}c_{i}!}\right)
\prod_{\ell=1}^{\infty}\prod_{j=0}^{\ell-1}
\left(\frac{j+\frac{\alpha}{Qr}}{j+1}\right)^{c_{\ell}}.
%
\]
One can verify that this distribution satisfies all the required
transition equations.  Yet, contrary to the sBDI model, it is more
difficult to determine the distributions of $C$, $c_{k}$ and $n_{i}$
based on this formulation; in particular the factor
$Q!/\left(Q-C\right)!$ prevents us from applying the same mathematical
procedure used in the sBDI case.

We can however take a different route, namely invoking neutrality and
the independence of the system, to deduce the distributions of $C$ and
$c_{k}$.  Since each species behaves independently from all others, we
can consider the number $m_{i}$ of individuals in the $i^{\rm th}$ species
(with $1\leq i\leq Q$) independently from the rest. Note that $m_{i}$
is a random variable that can be zero when there are no individuals of
species $i$ present in the system. The quantity $m_i$ is the
counterpart to $n_{i}$ introduced for the sBDI model with the caveat
that $n_i$ represents the number of individuals of a species
\emph{actually present} on the island (i.e. $\P{n_{i}=0}=0$).
In the current model $n_{i}$ can be expressed as a function of $m_i$ via
\begin{equation}
\P{n_{i}=k}=\P{m_{i}=k|m_{i}>0} \qquad \text{for }k\geq1, \label{eq:def n_i}
\end{equation}

\noindent
describing the distribution of the $i^{\rm th}$ species provided that at
least one of its individuals is on the island.  The random variable
$m_{i}$ follows a birth and death process: its birth rate is
$\alpha/Q+rm_{i}$ and its death rate is $\mu m_{i}$.  The $\alpha/Q$
rate corresponds to immigration, the rate $r m_{i}$ corresponds to actual
reproduction. We already determined the steady state distribution of
this process in Eq.~(\ref{eq:mod1 PN}), yielding a negative binomial
distribution \added{with} parameters $\alpha/(rQ)$ and $r/\mu$ as follows
\[
P(m_{i})=\left(1-\frac{r}{\mu}\right)^{\alpha/\left(Qr\right)}
\left(\frac{r}{\mu}\right)^{m_{i}}\frac{1}{m_{i}!}\prod_{k=0}^{m_{i}-1}
\left(\frac{\alpha}{Qr}+k\right).
\]
The $P(n_i)$ distribution can be determined from $P(m_i)$ expressed
above, using Eq.~(\ref{eq:def n_i})
\[
\P{n_{i}=k}=\frac{\P{m_{i}=k}}{1-\P{m_{i}=0}}=
\frac{\left(1-\frac{r}{\mu}\right)^{\alpha/\left(Qr\right)}}
{1-\left(1-\frac{r}{\mu}\right)^{\alpha/\left(Qr\right)}}
\left(\frac{r}{\mu}\right)^{k} \frac{1}{k!}\prod_{k'=0}^{k-1}
\left(\frac{\alpha}{Qr}+k'\right) \qquad \text{for any }k\geq1. 
\]

\noindent
Finally, the number of species $c_{k}$ with $k$ individuals and the
total number of species $C$ can be expressed as a function of $m_{i}$
as follows
\[
c_{k}=\sum_{i=1}^{Q}I\left(m_{i}=k\right)\quad\text{and}\quad 
C=\sum_{i=1}^{Q}I\left(m_{i}>0\right).
\]
Since all $m_{i}$ are i.i.d., the probability distributions of $c_{k}$
and $C$ are given by 
\begin{align*}
P(c_{k}) & =\binom{Q}{c_{k}}\P{m_{i}=k}^{c_{k}}\left(1-\P{m_{i}=k}\right)^{Q-c_{k}},\\
P(C) & =\binom{Q}{C}\left(1-\P{m_{i}=0}\right)^{C}\P{m_{i}=k}^{Q-C},
\end{align*}
which are binomial distributions of respective parameters $Q$ and
$\P{m_{i}=k}$ for $c_{k}$, and $Q$ and $1-\P{m_{i}=0}$ for $C$.  Note
that this approach does not allow us to determine the diversity
indices $H$ and $S$.}


\subsection{\label{subsec:S-mod1-Convergence high immigration}Convergences in
the large immigration regime}

In this section, we will prove the convergence of
\[
N/\Omega,\quad C/\Omega,\quad\left(\frac{c_{1}}{\Omega},
\frac{c_{2}}{\Omega},\ldots\right),\quad\text{and}\quad H/\log\Omega
\]
in the large immigration regime defined by $\alpha = \widetilde{\alpha}\Omega$, 
$\Omega \to \infty$.
\begin{prop}
\label{prop:Model1_scaled_N}The scaled total number of individuals
$N/\Omega$ converges in distribution to the constant
$\widetilde{\alpha}/(\mu-r)$.
\end{prop}

\begin{proof}
The definition of the convergence in distribution described in
Equation~(\ref{eq:def convergence distrib}) is equivalent to the
convergence of its moment generating function. One is left 
with showing that 

\[
\text{for any }\xi<0,\quad\lim_{\Omega\rightarrow\infty}\E{e^{\xi N/\Omega}}
=\frac{\widetilde{\alpha}}{\mu-r}
\]
(see for instance \citet[Chapter 5]{Billingsley2012}). Since
$N\sim\text{NegBinom}\left(\widetilde{\alpha}\Omega/r,r/\mu\right)$
for which the moment generating function is known, we have for any $\xi<0$:
\begin{align*}
\E{e^{\xi N/\Omega}} & =
\left(\frac{1-r/\mu}{1-e^{\xi/\Omega}r/\mu}\right)^{\widetilde{\alpha}\Omega/r}.
\end{align*}
Upon taking the logarithm of the previous expression, we find
\begin{align*}
\log\left[\E{e^{\xi N/\Omega}}\right] &
=\frac{\widetilde{\alpha}\Omega}{r}
\left[\log\left(1-r/\mu\right)-\log\left(1-e^{\xi/\Omega}r/\mu\right)\right]\\ &
\simlim{\Omega\to\infty}-\frac{\widetilde{\alpha}\Omega}{r}
\log\left[1-\frac{\xi}{\Omega}\frac{r/\mu}{1-r/\mu}\right]\\ &
\simlim{\Omega\to\infty}\frac{\widetilde{\alpha}\Omega}{r}
\frac{\xi}{\Omega}\frac{r/\mu}{1-r/\mu} = \xi\frac{\widetilde{\alpha}}{\mu-r},
\end{align*}
so
\[
\E{e^{\xi N_{\Omega}/\Omega}}\xrightarrow
[\Omega\rightarrow\infty]{}\exp\left[\xi\frac{\widetilde{\alpha}}{\mu-r}\right],
\]
thus proving the proposition.
\end{proof}
\begin{prop}
The scaled total number of species $C/\Omega$ converges
in distribution to
\[
{C\over \Omega}\xrightarrow[\Omega\rightarrow\infty]{{\cal D}}
\frac{\widetilde{\alpha}}{r}\log\left[\frac{1}{1-r/\mu}\right].
\]
\end{prop}

\begin{proof}
The proof is similar to Proposition~\ref{prop:Model1_scaled_N}.
\end{proof}

\begin{prop}
For each $k>0$,  $c_{k}/\Omega$ converges
in distribution to
\[
{c_{k}\over \Omega}\xrightarrow[\Omega\rightarrow\infty]{{\cal D}}
\frac{\widetilde{\alpha}}{r}\,\frac{(r/\mu)^{k}}{k}.
\]
\end{prop}

\begin{proof}
For any vector $\vec{c}$ and $k\geq1$, we have that
\[
c_{k}=\sum_{i=1}^{C}\ind{n_{i},k}.
\]
Consider the moment generating function of the random variable
$c_{k}$. For any $\xi<0$, we have
\[
\E{e^{\xi c_{k}/\Omega}}=\E{\exp
\left(\frac{\xi}{\Omega}\sum_{i=1}^{C}\ind{n_{i},k}\right)}.
\]
\added{Since} $n_{i}$ are identical and independently distributed
and independent of $C$, and since their distributions do
not depend on the parameter $\Omega$, it follows that

\begin{align*}
\E{e^{\xi c_{k}/\Omega}} & 
=\E{\left(\E{\exp\left(\frac{\xi}{\Omega}\ind{n_{1},k}\right)}\right)^{C}}\\
 & =\E{\left(e^{\xi/\Omega} P(n_{1}=k) +
\left(1-P(n_{1}=k)\right)\right)^{C}}\\
 & =\E{\left(\left(e^{\xi/\Omega}-1\right)P(n_{1}=k)+1\right)^{C}}.
\end{align*}
Since the probability distribution of $n_{1}$ is known, we have

\begin{align*}
\E{e^{\xi c_{k}/\Omega}} & =\E{\left(1-{1\over k}\left({r \over \mu}\right)^{k}\,
\frac{(e^{\xi/\Omega}-1)}{\log(1-r/\mu)}\right)^{C}}.
\end{align*}
Note that for any real $A$,
\begin{align*}
C\log\left[1-\left(e^{\xi/\Omega}-1\right)A\right] & 
\simlim{\Omega\to\infty}-C\left(e^{\xi/\Omega}-1\right)A,\\
 & \simlim{\Omega\to\infty}-{C\over \Omega}\xi A.
\end{align*}
Considering the exponential of this expression, we have
\[
\E{e^{\xi c_{k}/\Omega}}=\E{\exp\left[-{C\over \Omega}\frac{(r/\mu)^{k}}{k}\,
\frac{\xi}{\log(1-r/\mu)}\right]}.
\]
Finally, since we have already shown that $C/\Omega$ converges
in distribution (Proposition 2 above), we find
\[
\lim_{\Omega\rightarrow\infty}\E{e^{\xi c_{k}/\Omega}}=
\exp\left(\xi\frac{\widetilde{\alpha}}{r}\,\frac{(r/\mu)^{k}}{k}\right).
\]
\end{proof}

\begin{prop}
The Shannon's Entropy $H$ converges in distribution 
\added{as} 
\[
{H\over \log\Omega}\xrightarrow[\Omega\rightarrow\infty]{\mathbb{{\cal D}}}1.
\]
\end{prop}

\begin{proof}
\added{Using} the definition of $H$,
\[
{H\over \log\Omega}=\sum_{k=1}^{\infty}k\,\frac{c_{k}}{\Omega}\,
\frac{\Omega}{N}\frac{\log N-\log k}{\log\Omega},
\]
where $c_{k}/\Omega$ and $N/\Omega$ converge
in distribution to known constants, we find
\[
{H\over \log\Omega}\xrightarrow[\Omega\rightarrow\infty]{\mathbb{{\cal D}}}
\frac{\mu-r}{r}\sum_{k=1}^{\infty}\left({r\over \mu}\right)^{k}=1
\]
\end{proof}
\begin{prop}
The Simpson's diversity index $S$ converges in distribution \added{as}
\[
S \xrightarrow[\Omega\rightarrow\infty]
{\mathbb{{\cal D}}} 1.
\]
\end{prop}

\begin{proof}
By the definition of $S$ (Equation~\ref{eq:def diversity indices})
\[
S=1 - \frac{1}{\Omega}\sum_{k=1}^{\infty}\frac{c_{k}}{\Omega}\left(\frac{k}{N/\Omega}\right)^{2},
\]
and \added{since} $c_{k}/\Omega$ and $N/\Omega$ converge in distribution to known 
constants, we find
\[
S \xrightarrow[\Omega\rightarrow\infty]{\mathbb{{\cal D}}}
 - \frac{1}{\Omega}
\sum_{k=1}^{\infty}\frac{\widetilde{\alpha}}{r}\,k\left(\frac{r}{\mu}\right)^{k}
\left(\frac{\mu-r}{\widetilde{\alpha}}\right)^{2}
= 1 - \frac{\left(\mu-r\right)^{2}}{\Omega r\widetilde{\alpha}}\sum_{k=1}^{\infty}k\,\left(\frac{r}{\mu}\right)^{k}.
\]
One can then recognize the power series identity
\[
\frac{r/\mu}{(1-r/\mu)^{2}}=\sum_{k=1}^{\infty}k\,\left(\frac{r}{\mu}\right)^{k}
\]
and hence show that the second term vanishes as $\Omega \to \infty$
and deduce the result $S \xrightarrow[\Omega\rightarrow\infty]
{\mathbb{{\cal D}}} 1$.

\end{proof}

\section{BDI model with mutation (BDIM)}

\begin{figure}
\includegraphics[width=1\textwidth]{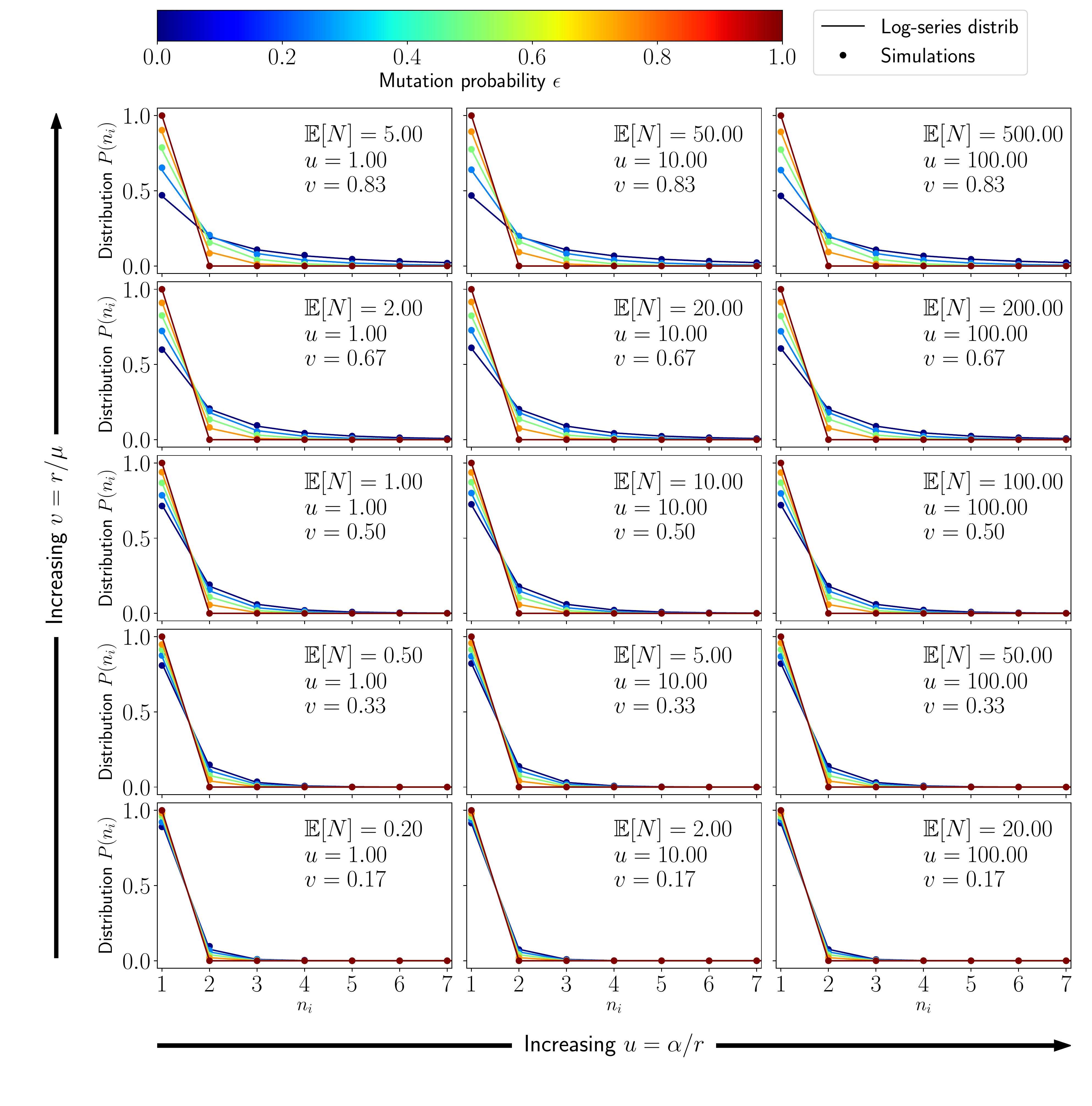}
\caption{\label{fig:S-mod2: Sim distrib n_i}Distribution of the number
  of individuals in one species $n_{i}$ \added{under} different
  parameter \added{choices.} \added{Dots represent simulations for
    various values of $u=\alpha/r$, $v=r/\mu$, ($r=1$) and $\epsilon$;
    solid lines depict logarithmic distributions with parameter
    $r(1-\epsilon)/\mu$.} As expected, the logarithmic distributions
  match the simulations $n_{i}$, and the distributions of $n_{i}$ do
  not depend on $u$.}
\end{figure}

\begin{figure}
\includegraphics[width=1\textwidth]{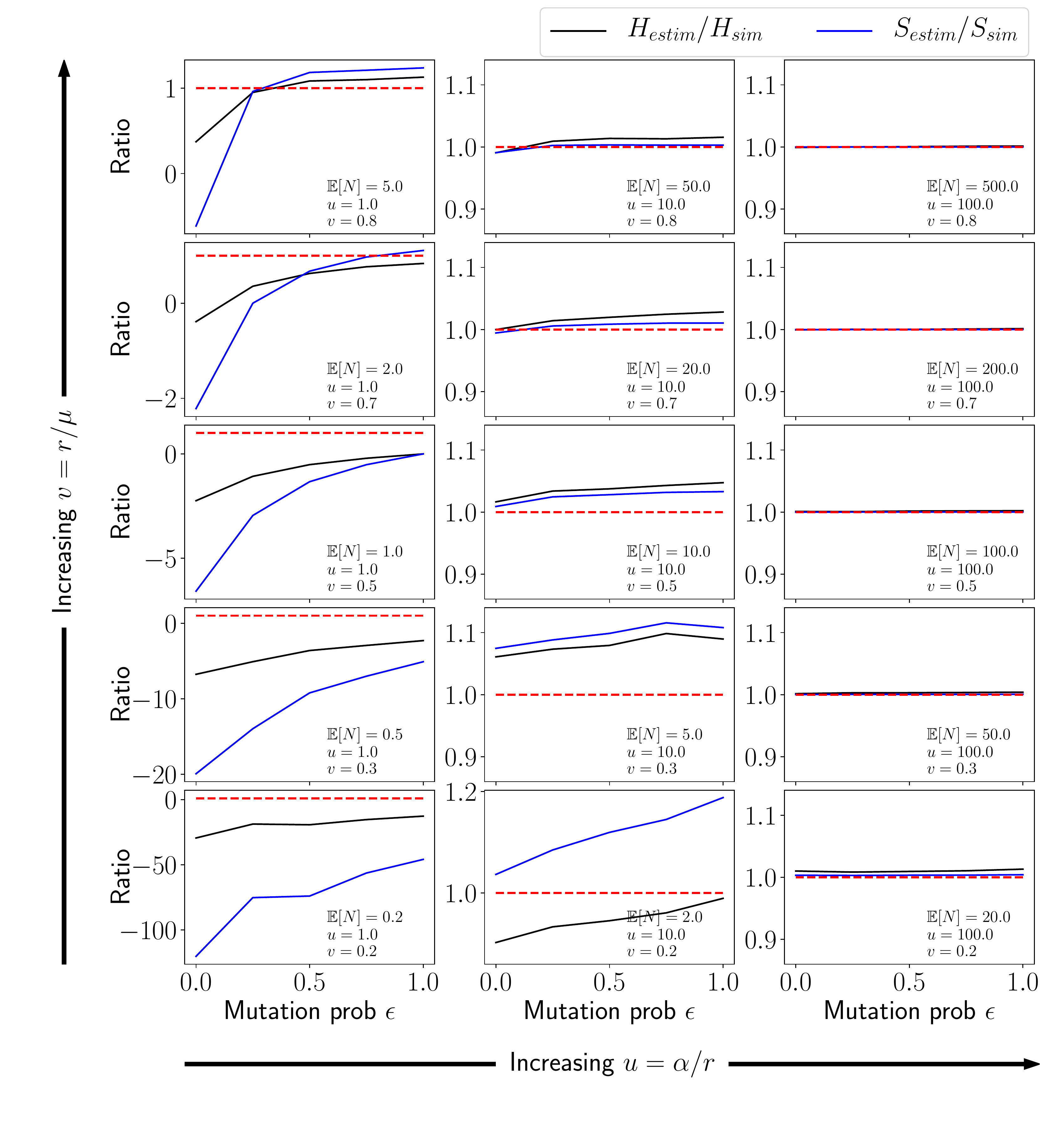}
\caption{\label{fig:S-mod2: Estim HS} Accuracy of Shannon's entropy
  and Simpson's index (as defined in Equation~(\ref{eq:mod2 EH ES})).
  We plot the ratio of the estimates of Shannon's entropy and
  Simpson's index and their respective values measured via simulation
  for different $u=\alpha/r$, $v=r/\mu$ (by taking $r=1$), and
  different $\epsilon$. The estimates become more accurate as $\E N$
  increases: the error is below $10\%$ for any parameters $u$, $v$,
  $\epsilon$ such that $\E N$ is larger than $5$.}
\end{figure}

\subsection{\label{subsec:S-mod2 distrib n_i}Distribution of the number of individual
in one species}

We propose an argument for a Log-series distribution of any species
\[
\pi_{k}=P(n_{i}=k)
\]
when all species are independent of each other. There are several ways
to interpret $\pi_{k}$. First consider the explicit dynamics of each
species. Denote by $m_{q}(t)$ the number of individuals of
species $q$ at time $t$ and define $a_{q}$ as the time of arrival (by
convention, we order the species such as $a_{0}=0<a_{1}<a_{2}<\ldots$)
and $d_{q}$ its ``lifespan'', i.e. the species will be extinct at time
$a_{q}+d_{q}$ (see the example \deleted{on}\added{in} Figure~\ref{Fig10}(a)).  Note that the
index $q$ indicates the order of arrival (and not the species
identity index $i$ used in the main article), and that the
distribution of the \added{times} $a_{q}$ is not specified and can be adapted to any
rate of species creation (either by immigration or by mutation). 
The evolution of each species is independent of each other, and each
of them defines an identically distributed birth-death process
characterized by the following transitions

\begin{equation}
\begin{cases}
m_{q}\rightarrow m_{q}+1 & \text{at rate }m_{q}\,r(1-\epsilon),\\
m_{q}\rightarrow m_{q}-1 & \text{at rate }m_{q}\,\mu.
\end{cases}\label{eq:S-mod2 n dynamic-2}
\end{equation}
Due to the $r<\mu$ assumption, this process will become extinct almost
surely \citep[Chapter 2]{Bansaye2015a} and the lifespan $d_{q}$ of
each species is finite.

\begin{figure}
\begin{center}
\includegraphics[width=5in]{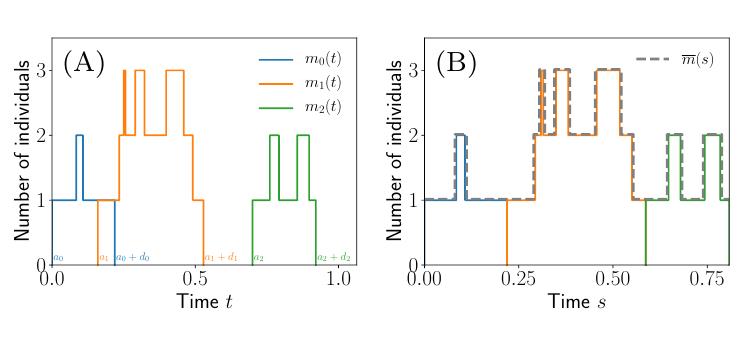}\vspace{-4mm}
\caption{(a) A representative trajectory of three immigrated species.
  The $q$-th species is introduced at time $a_{q}$ and extinguishes at
  time $a_{q}+d_{q}$. (b) Construction of the process $\overline{m}$
  (defined in Equation~(\ref{eq:S_mod2 define overline m})) by stacking and 
  concatenating the trajectories of each species $\left(m_{q}\right)_{q\in\protect\N}$.
\label{Fig10}}
\end{center}
\end{figure}

In the main article, we interpreted $\pi_{k}$ as the
number of individuals in a given species at steady state, that is to say,
we considered the $T\rightarrow\infty$ limit
\[
\pi_{k}=\lim_{T\rightarrow\infty}\P{m_{J_{T}}(T)=k}
\]
where $J_{T}$ is the index of a randomly sampled species among those
that \deleted{existing}\added{exist} at time $T$; \deleted{that is to
  say}\added{i.e.}, $J_{T}$ is uniformly chosen among all the species
$q$ such that $a_{q}<T<a_{q}+d_{q}$.

However, there is another way to interpret $\pi_{k}$.  Consider all
species that \emph{exist or have existed} up to \deleted{the} time $T$ and
then randomly select one of them,  species $I_{T}$.  The number of 
individuals in species
$I_{T}$ at a\deleted{n} randomly chosen time $\tau_{I_{T}}$ between the
introduction of the species (at time $a_{I_{T}}$) and the extinction
(at time $a_{I_{T}}+d_{I_{T}}$) is denoted $m_{I_{T}}$. In this
picture, we can characterize $\pi_{k}$ according to
\begin{equation}
\pi_{k}=\lim_{T\rightarrow\infty}\P{m_{I_{T}}(\tau_{I_{T}})=k}.
\label{eq:S_mod2 pi_k}
\end{equation}

The main difference between the two approaches is that, in the first
case, we sample among the species that exist\deleted{s} at a precise time $T$ 
before taking $T\to \infty$, while in the second case,
we sample among all the species that existed before time $T$ 
(before taking $T\to \infty$).

For a fixed time $T$, the last species introduced in the system
is given by
\[
Q_{T}=\argmax_{q\in\N}\left(a_{q}<T\right).
\]
All species that exist\deleted{s} or have existed before time $T$ are in
the set $\left\{0,\ldots,Q_{T}\right\} $. Note that \deleted{as}\added{since} $a_{q}$ are
increasing in $q$, $\lim_{T\rightarrow\infty} Q_{T}=\infty$.  As per
Equation~(\ref{eq:S_mod2 pi_k}), we have to sample one species among
the set $\left\{0,\ldots,Q_{T}\right\} $.  One key point is that the
random selection is not uniform: there is a higher chance of selecting
species with longer lifespans.  If $I_{T}$ is the index of the
randomly chosen species, we can write

\[
\P{I_{T}=q}=\ind{q\leq Q_{T}}\frac{d_{q}}{\sum_{j=0}^{Q_{T}}d_{j}}.
\]
The first term $\ind{q\leq Q_{T}}$ ensures that the species
$q$ exists before time $T$ while the second term proportionally weights 
the probability of sampling according to their lifespans. 
Conditioned on species  $I_{T}$ having been sampled, we then 
randomly chose a time $\tau_{I_{T}}$ uniformly distributed between $a_{I_{T}}$
and $a_{I_{T}}+d_{I_{T}}$.

\begin{prop}
The limiting distribution becomes
\[
\pi_{k}=\lim_{T\rightarrow\infty}\P{m_{I_{T}}(\tau_{I_{T}})=k}=
\frac{1}{\log\left(1-\frac{r(1-\epsilon)}{\mu}\right)}
\frac{1}{k}\left(\frac{r(1-\epsilon)}{\mu}\right)^{k}.
\]
\end{prop}

\begin{proof}
By summing over all possible species $q$, we can write
\begin{align*}
\P{m_{I_{T}}(\tau_{I_{T}})=k} & =\sum_{q\in\N}\E{\ind{q\leq Q_{T}}
\frac{d_{q}}{\sum_{j=0}^{Q_{T}}d_{j}}\ind{m_{q}(\tau_{q}),k}}\\
 & =\E{\sum_{q=0}^{Q_{T}}\frac{d_{q}}{\sum_{j=0}^{Q_{T}}d_{j}}
\frac{1}{d_{q}}\int_{a_{q}}^{a_{q}+d_{q}}\ind{m_{q}(t),k}\,{\rm d}t}\\
 & =\E{\frac{\sum_{q=0}^{Q_{T}}\int_{a_{q}}^{a_{q}+d_{q}}\ind{m_{q}(t),k}\,{\rm d}t}{\sum_{j=0}^{Q_{T}}d_{j}}}
\end{align*}
Next, consider the process $\overline{m}(s)$ defined as
\begin{equation}
\overline{m}(t)=m_{\nu_{t}}\left(t-\overline{d}_{\nu(t)}+a_{\nu(t)}\right)\label{eq:S_mod2 define overline m}
\end{equation}
with
\[
\overline{d}_{k}=\sum_{q=1}^{k-1}d_{q}\quad\text{and}\quad\nu(t)=\argmax_{q}\left\{ \sum_{j=0}^{q}d_{j}<t\right\} .
\]
The process $\overline{m}$ is simply the stacking of all the processes
$m_{q}$ in the sense that the process $\overline{m}(t)$ for $t$
between $\overline{d}_{q}$ and $\overline{d}_{q+1}$ will be equal to
the process $m_{q}(s)$ for $s=t-\overline{d}_{q}+a_{q}$ between
$a_{q}$ and its extinction time $a_{q}+d_{q}$ (see the example on
Figure~\ref{Fig10}(b)). With this stacked process,

\begin{align*}
\P{m_{I_{T}}(\tau_{I_{T}})=k} & =\E{\frac{\int_{0}^{a_{\delta_{T}}+d_{\delta_{T}}}
\ind{\overline{m}(t),k}\,{\rm d}t}{\overline{d}_{Q_{T}}}}.
\end{align*}
By ergodicity of the process $\overline{m}$, we have
\[
\lim_{T\rightarrow\infty}\P{m_{I_{T}}(\tau_{I_{T}})=k}=\lim_{T\rightarrow\infty}\P{\overline{m}(T)=k}.
\]
Finally, we have to determine the steady state of the process $\overline{m}$.
Since the transitions of the process $\overline{m}$ are a simple birth-death process
\begin{equation}
\begin{cases}
\overline{m}\rightarrow\overline{m}+1 & \text{at rate }\quad\overline{m}\,r(1-\epsilon)\\
\overline{m}\rightarrow\overline{m}-1 & \text{at rate }\quad\overline{m}\,\mu
\ind{\overline{m}>0}.
\end{cases},\label{eq:S-mod2 n dynamic-1-1}
\end{equation}
we have that its equilibrium distribution is a logarithmic series
distribution \deleted{of}\added{with} parameter $p\equiv
r(1-\epsilon)/\mu$ (by imposing equations of detailed balance).
\end{proof}

\subsection{\label{subsec:S-mod2 Moments of C}Moments of $C$}

The third relation of Equation~(\ref{eq:relations NCc}) 
yields the following expression for the moment generating function of $N$:
\[
\E{e^{\xi N}}=\E{\prod_{i=1}^{C}e^{\xi n_{i}}},
\]
for any $\xi<0$. Since all the $\left(n_{i}\right)_{i\leq C}$ are
identical and independently distributed and independent of $C$, we
have
\begin{align}
\E{e^{\xi N}} & =\E{\E{e^{\xi n_{1}}}^{C}} = 
\E{\left(\frac{\log\left(1-pe^{\xi}\right)}{\log\left(1-p\right)}\right)^{C}}.
\label{MGC}
\end{align}
\deleted{since} Equation~(\ref{eq:mod2 Pni}) shows that the distribution over
$n_{1}$ is a log-series distribution \deleted{of}\added{with} parameter
$p=r\left(1-\epsilon\right)$.  By redefining the variable $\xi'$ such
that $e^{\xi'}:=\log\left(1-pe^{\xi}\right)/\log\left(1-p\right)$ and
eliminating $\xi$ for $\xi'$, Equation~(\ref{MGC}) becomes an
expression for the moment generating function of $C$,
\[
\E{e^{\xi'C}}=\left(\frac{1-r/\mu}{1-\frac{r}{\mu}
\frac{1-\left(1-p\right)^{e^{\xi'}}}{p}}\right)^{\alpha/r}.
\]
\added{By differentiating this expression, we can determine the second moment of $C$:}

\begin{align*}
\E{C^{2}} & =\lim_{\xi'\rightarrow0}\frac{{\rm d}^{2}}{{\rm d}\xi'^{2}}\E{e^{\xi'C}}=
\E{C}\left[1+\log\left(1-p\right)+ \left(1+\frac{r}{\alpha}\right)\E{C}\right],
\end{align*}
which yields the expression for $\var{C}$ in Equation~(\ref{eq:mod2 VC}).

\section{BDI model with carrying capacity (BDICC)}

\subsection{\label{subsec:S-mod3-Distrib N}Steady state distribution of $\vec{c}$}

To determine $P(\vec{c})$, the probability of occurrence of the
species-count state $\vec{c}$, first consider a finite
$K=\argmax_{i}(c_{i}>0)$. As explained in the main text, if the
system is reversible, one instance of Equation~(\ref{eq:mod3 balance_ck}) 
is
\[
\mu(N)c_{K}K P(\vec{c})=
(K-1)\left(c_{K-1}+1\right)r P(c_{1},\ldots,c_{K-1}+1,c_{K}-1,\vec{0}).
\]
Recursively unwinding this relationship, we find
\begin{align*}
P(\vec{c}) & =P(c_{1},\ldots,c_{K-1}+1,c_{K}-1,\vec{0})
\frac{r}{\mu(N)}\frac{K-1}{K}\frac{c_{K-1}+1}{c_{K}},\\
P(\vec{c}) & =P(c_{1},\ldots,c_{K-1}+c_{K},0,\vec{0})
\frac{r^{c_{K}}}{\mu(N)\ldots\mu(N-c_{K}+1)}\left(\frac{K-1}{K}\right)^{c_{K}}
\frac{\left(c_{K-1}+c_{K}\right)!}{c_{K}!c_{K-1}!},\\
P(\vec{c}) & =P(C,\vec{0})\frac{r^{N-C}}{\mu(N)\ldots\mu(N-(K-1)c_{K}-\ldots-c_{2}+1)}
\prod_{i=1}^{K-1}\prod_{j=i+1}^{K}\left(\frac{i}{i+1}\right)^{c_{j}}
\frac{C!}{\prod_{i=1}^{K}c_{i}!},\\
P(\vec{c}) & =P(C,\vec{0})\frac{r^{N-C}}{\prod_{n=1}^{N-C}\mu(N-n+1)}
\frac{C!}{\prod_{i=1}^{K}i^{c_{i}}c_{i}!}.
\end{align*}
After applying Equation~(\ref{eq:mod3 balance2_ck}), we have
by recursion
\[
P(C,0,\ldots)=\frac{\alpha}{\mu(C)}\frac{1}{C}P(C-1,0,\ldots)
=\frac{\alpha^{C}}{C!\prod_{i=1}^{C}\mu(i)}P(0,\ldots),
\]
and 
\[
P(\vec{c})=P(0,\ldots)\left(\frac{\alpha}{r}\right)^{C}
\frac{r^{N}}{\prod_{n=1}^{N}\mu(n)}\frac{1}{\prod_{i=1}^{K}i^{c_{i}}c_{i}!}.
\]
Since the state $\vec{c}=\vec{0}$ uniquely corresponds to the state
$N=0$ and the \added{above} expression holds for $K$ arbitrarily
large, it follows that

\begin{equation}
P(\vec{c})=\frac{1}{Z_{\alpha,r,\mu}}\left(\frac{\alpha}{r}\right)^{C}
\frac{r^{N}}{\prod_{n=1}^{N}\mu(n)}\frac{1}{\prod_{i=1}^{\infty}i^{c_{i}}c_{i}!}.
\label{eq:mod3_P(c)}
\end{equation}

One can verify that this steady-state distribution satisfies the detailed balanced 
conditions connecting all pairs of states:

\begin{equation}
\begin{cases}
\mu\left(\sum_{k}kc_{k}\right) kc_{k}P(\vec{c})=
(k-1)\left(c_{k-1}+1\right) r P(c_{1}, \ldots, c_{k-1}+1,c_{k}-1, \ldots) & \forall k>1, \\
\mu\left(\sum_{k}kc_{k}\right) c_{1}P(\vec{c})=\alpha P(c_{1}-1,\ldots, c_{k}, \ldots).
\end{cases}
\label{eq:S_mod3 verfying P vec c}
\end{equation}


\subsection{\label{subsec:S-mod3-Convergence N_Omega}Convergence of $N/\Omega$}
\begin{thm}
\label{thm:mod3_N}The random variable $N/\Omega$ converges
in probability to the real $n^{*}$ which is the only solution of
the fixed point Equation~(\ref{eq:fixed-pt-eq}).
\end{thm}
To prove this Theorem,  first define
\[
f(x):=\frac{\widetilde{\alpha}+rx}{x\widetilde{\mu}(x)}\quad\text{and}
\quad f_{k}:=
\frac{\widetilde{\alpha}+(k-1)r/\Omega}{(k/\Omega) \widetilde{\mu}(k/\Omega)}
\,\quad\forall k\in\N_{*}.
\]
The function $f$ defines the steady-state constraint on $n=N/\Omega$
given by Equation~(\ref{eq:fixed-pt-eq}) where $x=n^{*}$ is the 
only real solution to $f(x)=1$  With these definitions, the
probability distribution over $N$ can be expressed as
\[
\forall n\in\N,\qquad\P{N=n}=
\frac{\exp\left(\sum_{k=1}^{n}\log f_{k}\right)}
{\sum_{n'=0}^{\infty}\exp\left(\sum_{k=1}^{n'}\log f_{k}\right)}.
\]


\noindent Now, consider the following lemma:
\begin{lem}
\label{lem:Decrease_f_x}The function $f$ is strictly decreasing
and there exists a $\Omega^{*}$ for which $\,\forall\Omega\geq\Omega^{*}$,
$\left(f_{k}\right)_{k\geq1}$ is a decreasing sequence.
\end{lem}

\begin{proof}
The decrease of the function $f$ is a direct implication of the increase
of $\widetilde{\mu}$. For, $\left(f_{k}\right)_{k\geq1}$ we have
\[
\left(k+1\right)\left(\widetilde{\alpha}\Omega+r(k-1)\right)-k\,
\left(\widetilde{\alpha}\Omega+r k\right)=\widetilde{\alpha}\Omega-r,
\]
which is positive for large enough $\Omega$. Since $\widetilde{\mu}$
is increasing,
\[
\frac{f_{k}}{f_{k+1}}=
\frac{\widetilde{\mu}(\left(k+1\right)/\Omega)}{\widetilde{\mu}(k/\Omega)}
\frac{\left(k+1\right)\left(\widetilde{\alpha}\Omega+(k-1)r\right)}
{k\,\left(\widetilde{\alpha}\Omega+r k\right)}>1.
\]
\end{proof}
To prove Theorem~\ref{thm:mod3_N}, we have to show that $\forall\delta >0$,
\[
\P{\left|N/\Omega-n^{*}\right|>\delta}\xrightarrow[\Omega\rightarrow\infty]{}0
\]
that is to say, we have to show that
\begin{align}
\P{N/\Omega>n^{*}+\delta} & \xrightarrow[\Omega\rightarrow\infty]{}0,\label{eq:cvg1}\\
\P{N/\Omega<n^{*}-\delta} & \xrightarrow[\Omega\rightarrow\infty]{}0.\label{eq:cvg2}
\end{align}
The proofs of convergence for both limits above are very similar so we
will focus on the proof of Equation~(\ref{eq:cvg1}). To simplify
notation, we define $a_{\Omega,\delta} \equiv \left\lceil
\Omega\left(n^{*}+\delta \right)\right\rceil$, (where
$\left\lceil\cdot \right\rceil$ is the ceiling function). Since the
distribution of $N$ is known, we have
\begin{align*}
\P{N/\Omega> n^{*}+\delta} & =
\frac{\sum_{n=a_{\Omega,\delta}}^{\infty}
\exp\left(\sum_{k=1}^{n}\log f_{k}\right)}{\sum_{n=0}^{a_{\Omega,\delta}-1}
\exp\left(\sum_{k=1}^{n}\log f_{k}\right)+\sum_{n=a_{\Omega,\delta}}^{\infty}
\exp\left(\sum_{k=1}^{n}\log f_{k}\right)}\\
 & =\left(\frac{\sum_{n=0}^{a_{\Omega,\delta}-1}
\exp\left(\sum_{k=1}^{n}\log f_{k}\right)}{\sum_{n=a_{\Omega,\delta}}^{\infty}
\exp\left(\sum_{k=1}^{n}\log f_{k}\right)}+1\right)^{-1}.
\end{align*}
Thus, it is enough to show
\[
\frac{\sum_{n=0}^{a_{\Omega,\delta}-1}
\exp\left(\sum_{k=1}^{n}\log f_{k}\right)}{\sum_{n=a_{\Omega,\delta}}^{\infty}
\exp\left(\sum_{k=1}^{n}\log f_{k}\right)}\xrightarrow[\Omega\rightarrow\infty]{}\infty
\]
in order to prove the convergence of Equation~(\ref{eq:cvg1}).
\begin{prop}
In the $\Omega\rightarrow\infty$ limit, the following equivalence holds
\[
\sum_{n=a_{\Omega,\delta}}^{\infty}
\exp\left(\sum_{k=1}^{n}\log f_{k}\right)
\simlim{\Omega\to\infty}\exp\left(\sum_{k=1}^{a_{\Omega,\delta}-1}
\log f_{k}\right)\frac{1}{1-f(n^{*}+\delta)}
\]
\end{prop}

\begin{proof}
We first decompose the sum according to
\begin{align*}
\sum_{n=a_{\Omega,\delta}}^{\infty}\exp\left(\sum_{k=1}^{n}
\log f_{k}\right) & =\exp\left(\sum_{k=1}^{a_{\Omega,\delta}-1}
\log f_{k}\right)\sum_{n=a_{\Omega,\delta}}^{\infty}
\exp\left(\sum_{k=a_{\Omega,\delta}}^{n}\log f_{k}\right).
\end{align*}
The second term of the decomposition can be rewritten as
\[
\sum_{n=a_{\Omega,\delta}}^{\infty}
\exp\left(\sum_{k=a_{\Omega,\delta}}^{n}\log f_{k}\right)=
\sum_{n=0}^{\infty}\exp\left(\sum_{k=0}^{n}\log f_{k+a_{\Omega,\delta}}\right).
\]
Since $a_{k,\Omega}/\Omega\xrightarrow[\Omega\rightarrow\infty]{}n^{*}+\delta$,
it follows that
\[
\sum_{k=0}^{n}\log f_{k+a_{\Omega,\delta}}
\simlim{\Omega\to\infty}n\,\log f(n^{*}+\delta).
\]
As $f$ is a strictly decreasing function (cf. Lemma~\ref{lem:Decrease_f_x}),
and since $n^{*}$ is the only point where $f(n^{*})=1$, it follows
that $f\left(n^{*}+\delta\right)<1$. Therefore, the sum over $n$
converges, and we have
\[
\sum_{n=a_{\Omega,\delta}}^{\infty}
\exp\left(\sum_{k=a_{\Omega,\delta}}^{n}\log f_{k}\right)
\simlim{\Omega\to\infty}\frac{1}{1-f(n^{*}+\delta)}
\]
\end{proof}
With the previous Proposition, it is enough to prove that the ratio
\[
\frac{\sum_{n=0}^{a_{\Omega,\delta}-1}
\exp\left(\sum_{k=1}^{n}\log f_{k}\right)}{\exp\left(\sum_{k=1}^{a_{\Omega,\delta}-1}
\log f_{k}\right)}=\sum_{n=0}^{a_{\Omega,\delta}-1}
\exp\left(-\sum_{k=n+1}^{a_{\Omega,\delta}-1}\log f_{k}\right)
\]
diverges to infinity in order to prove the convergence of Equation~(\ref{eq:cvg1}).
\begin{prop}
The sum
\[
\sum_{n=0}^{a_{\Omega,\delta}-1}\exp\left(-\sum_{k=n+1}^{a_{\Omega,\delta}-1}
\log f_{k}\right)\xrightarrow[\Omega\rightarrow\infty]{}\infty
\]
diverges.
\end{prop}

\begin{proof}
Since $\left(f_{k}\right)_{k\geq1}$ is decreasing for large $\Omega$
(cf. Lemma~\ref{lem:Decrease_f_x}), we have 
\[
\sum_{k=n+1}^{a_{\Omega,\delta}-1}\log f_{k}\leq\left(a_{\Omega,\delta}-n-1\right)
\log f_{a_{\Omega,\delta}-1}
\]
for sufficiently large $\Omega$. Therefore, 
\[
\sum_{n=0}^{a_{\Omega,\delta}-1}\exp\left(-\sum_{k=n+1}^{a_{\Omega,\delta}-1}
\log f_{k}\right)\geq\sum_{n'=1}^{a_{\Omega,\delta}}
\left(\frac{1}{f_{a_{\Omega,\delta}-1}}\right)^{n'}.
\]
Since
\[
f_{a_{\Omega,\delta}-1}\xrightarrow[\Omega\rightarrow\infty]{}f(n^{*}+\delta),
\]
for large enough $\Omega$ and since $f$ is decreasing, we have that
$f_{a_{\Omega,\delta}-1}<1-\eta$ for $\eta$ small enough. Therefore,
we conclude the divergence
\[
\sum_{n'=1}^{a_{\Omega,\delta}}\left(\frac{1}{f_{a_{\Omega,\delta}-1}}\right)^{n'}
\xrightarrow[\Omega\rightarrow\infty]{}\infty
\]
and proof of the proposition.
\end{proof}
With this Proposition, we have proven the convergence of
Equation~(\ref{eq:cvg1}).  The convergence of Equation~(\ref{eq:cvg2})
can be proved using exactly the same methods by considering
$b_{\Omega,\delta}= \left\lfloor
\Omega\left(\delta+n^{*}\right)\right\rfloor$ instead of
$a_{\Omega,\delta}$.

\subsection{\label{subsec:S_mod3 Convergence C/=0003A9}Convergence of $C/\Omega$}
\begin{thm}
The scaled total number of species $C/\Omega$ converges
in distribution to 
\[
\frac{C}{\Omega}\xrightarrow[\Omega\rightarrow\infty]{\cal{D}}
\frac{\widetilde{\alpha}}{r}\log\left[1+\frac{r}{\widetilde{\alpha}}n^{*}\right]=
\frac{\widetilde{\alpha}}{r}\log\left[\frac{1}{1-r/\widetilde{\mu}(n^{*})}\right],
\]
in which $n^{*}$ is the only real solution of the fixed point Equation~(\ref{eq:fixed-pt-eq}).
\end{thm}

\begin{proof}
One has to prove that
\[
\E{\exp\left[\xi C/\Omega\right]}=
\frac{Z_{\alpha e^{\xi/\Omega},r,\mu}}{Z_{\alpha,r,\mu}}
\xrightarrow[\Omega\to\infty]{}\left(\frac{1}{1-r/\widetilde{\mu}(n^{*})}\right)^
{\xi\widetilde{\alpha}/r}=\left(1+\frac{r}{\widetilde{\alpha}}
n^{*}\right)^{\xi\widetilde{\alpha}/r}
\]
with
\[
Z_{\alpha,r,\mu}=
\sum_{n'=0}^{\infty}\exp\left(\sum_{k=1}^{n'}
\log\frac{\widetilde{\alpha}+r(k-1)/\Omega}{k/\Omega\,\widetilde{\mu}(k/\Omega)}\right).
\]
First note that
\begin{align*}
\E{\exp\left[\xi C/\Omega\right]} & =
\frac{1}{Z_{\alpha,r,\mu}}\sum_{n=0}^{\infty}
\exp\left(\sum_{k=1}^{n}
\log\frac{\widetilde{\alpha}e^{\xi/\Omega}+r (k-1)/\Omega}{k/\Omega\,\widetilde{\mu}(k/\Omega)}\right)\\
 & =\sum_{n=0}^{\infty}\P{N=n}\exp\left(\sum_{k=1}^{n}
\log\frac{\widetilde{\alpha}e^{\xi/\Omega}+r(k-1)/\Omega}{\widetilde{\alpha}+r (k-1)/\Omega}\right)\\
 & =\E{\exp\left(\sum_{k=1}^{N}\log\frac{\widetilde{\alpha}e^{\xi/\Omega}+r (k-1)/\Omega}{\widetilde{\alpha}+r (k-1)/\Omega}\right)}
\end{align*}
Since $N/\Omega$ converges in probability to $n^{*}$,
\[
\E{\exp\left[\xi C/\Omega\right]}
\simlim{\Omega\to\infty}
\exp\left(\sum_{k=1}^{n^{*}\Omega}
\log\frac{\widetilde{\alpha}e^{\xi/\Omega}+r (k-1)/\Omega}{\widetilde{\alpha}+r(k-1)/\Omega}\right).
\]
%
%
\deleted{As}\added{Since} the function $\log\left(\frac{\widetilde{\alpha}e^{\xi/\Omega}+r
  \left(x-1\right)/\Omega}{\widetilde{\alpha}+r
  \left(x-1\right)/\Omega}\right)$ is decreasing in $x$, 
we can bound the sum with its lower and upper integral bounds
\begin{gather*}
\int_{1}^{n^{*}\Omega+1}\log
\frac{\widetilde{\alpha}e^{\xi/\Omega}+(x-1)r/\Omega}
{\widetilde{\alpha}+(x-1)r/\Omega}\,{\rm d}x\leq\sum_{k=1}^{n^{*}\Omega}\log
\frac{\widetilde{\alpha}e^{\xi/\Omega}+(k-1)r/\Omega}{\widetilde{\alpha}+(k-1)r/\Omega}\leq
\frac{\xi}{\Omega}+\int_{1}^{n^{*}\Omega}\log
\frac{\widetilde{\alpha}e^{\xi/\Omega}+(x-1)r/\Omega}
{\widetilde{\alpha}+(x-1)r/\Omega}\,{\rm d}x.
\end{gather*}
After rescaling $y = (x-1)/\Omega$, the bounds can be expressed as
\begin{gather*}
\Omega\int_{0}^{n^{*}+1/\Omega}\log
\frac{\widetilde{\alpha}e^{\xi/\Omega}+r y}{\widetilde{\alpha}+r y}\,{\rm d}y\leq
\sum_{k=1}^{n^{*}\Omega}\log\frac{\widetilde{\alpha}e^{\xi/\Omega}+ (k-1)r/\Omega}
{\widetilde{\alpha}+(k-1)r/\Omega}\leq\frac{\xi}{\Omega}+\Omega
\int_{0}^{n^{*}}\log\frac{\widetilde{\alpha}e^{\xi/\Omega}+r y}
{\widetilde{\alpha}+r y}\,{\rm d}y
\end{gather*}
Upon taking $\Omega \to \infty$ and expanding \added{the above
  expression, we find that} both bounds converge to
\[
\xi\int_{0}^{n^{*}}\frac{\widetilde{\alpha}}{\widetilde{\alpha}+r y}\,{\rm d}y.
\]
Thus, we find

\begin{align*}
\E{\exp\left[\xi C/\Omega\right]} & \simlim{\Omega\to\infty}
\exp\left(\xi\widetilde{\alpha}\,\int_{0}^{n^{*}}
\frac{1}{\widetilde{\alpha}+r u}\,{\rm d}u\right) = \exp\left(\xi\frac{\widetilde{\alpha}}{r}\,
\log\left[1+\frac{r}{\widetilde{\alpha}} n^{*}\right]\right).
\end{align*}
\end{proof}

\subsection{\label{subsec:S_mod3 Convergence ni}Convergence of $n_{i}$}
\begin{prop}
The marginal probability over each particle count $n_{i}$ 
converges according to 
\[
\P{n_{1}=k}\xrightarrow[\Omega\rightarrow\infty]{}\frac{1}{k}
\left(\frac{r}{\widetilde{\mu}(n^{*})}\right)^{k}\frac{-1}
{\log\left[1-r/\widetilde{\mu}(n^{*})\right]}.
\]
\end{prop}

\begin{proof}
\noindent \deleted{We have that}\added{The} $n_{i}$ values are identically distributed,
so that for any $i,j\leq C$, \deleted{we have that}
\[
\text{for any }k\geq1,\quad\P{n_{i}=k}=\P{n_{j}=k}.
\]

\noindent We can then compute the expectation
\[
\E{\frac{c_{k}}{C}}=\E{\sum_{i=1}^{C}\frac{\E{\ind{n_{i},k}|C}}{C}}
=\E{\E{\ind{n_{1},k}|C}}=\P{n_{1}=k}.
\]
This expectation is over a product of two converging quantities:
\[
\E{\frac{c_{k}}{C}}=\E{\frac{c_{k}}{\Omega}\,\frac{\Omega}{C}}= P(n_{1}=k),
\]
where $c_{k}/\Omega$ and $C/\Omega$ converge in distribution to constants
%
\[
\left(\frac{c_{k}}{\Omega},\frac{C}{\Omega}\right)
\xrightarrow[\Omega\rightarrow\infty]{{\cal D}}
\left(\frac{\widetilde{\alpha}}{r}
\log\left[\frac{1}{1-r/\widetilde{\mu}(n^{*})}\right],
\frac{\widetilde{\alpha}}{r}\,
\frac{1}{k}\left(\frac{r^{k}}{\widetilde{\mu}(n^{*})}\right)^{k}\right).
\]
We now apply the mapping theorem (see \citet[Chapter 5]{Billingsley2012}) to 
$\E{g\left(\frac{c_{k}}{\Omega},\frac{C}{\Omega}\right)}$ for
any continuous function $g$ to obtain
\[
P(n_{1}=k)\xrightarrow[\Omega\rightarrow\infty]{}
\frac{1}{k}\left(\frac{r}{\widetilde{\mu}(n^{*})}\right)^{k}
\frac{-1}{\log[1-r/\widetilde{\mu}(n^{*})]}.
\]
\end{proof}

\subsection{\label{subsec:S_mod3 BDICC on birth}Explicit breakdown of detailed balance 
in the BDICC-bis model with birth-mediated carrying capacity}

Here, we consider a Birth-Death-Immigration model with carrying
capacity but contrary to the BDICC model presented in
Figure~\ref{Fig1}(c), the carrying capacity is on the birth rate
$r(N)$, and the death rate $\mu$ is a constant. By analogy with the
BDICC analysis, we \added{find}\deleted{come up with} a sufficient condition for 
a steady state \added{to exist}
\[
\lim_{N\rightarrow\infty}r(N)<\mu.
\]
The distribution $P(N)$ of the total number of individuals is
given by
\[
P(N) =\begin{cases}
\displaystyle \frac{1}{Z_{\alpha,\mu}},\quad N=0,\\
\displaystyle \frac{1}{Z_{\alpha,\mu}}
\frac{1}{N!}{\displaystyle \prod_{k=0}^{N-1}}
\frac{\alpha+r(k) k}{\mu}, \quad N\geq 1,
\end{cases}
\]
where 
\[
Z_{\alpha,\mu}=1+\sum_{N=1}^{\infty}\frac{1}{N!}
{\displaystyle \prod_{k=0}^{N-1}}\frac{\alpha+r(k) k}{\mu}.
\]

All possible transitions of the BDICC-bis model are given by
\begin{align*}
 &  & (c_{1},c_{2},\ldots) & \xrightarrow{\alpha}(c_{1}+1,c_{2},\ldots) &  & \text{Immigration}\\
\text{for }k\geq1\quad &  & (c_{1},\ldots,c_{k},c_{k+1},\ldots) & \xrightarrow{r(N)kc_{k}}(c_{1},\ldots,c_{k}-1,c_{k+1}+1,\ldots) &  & \text{Birth}\\
\begin{aligned}\text{for }k\geq2\quad\\
\\
\end{aligned}
 &  & \begin{aligned}(c_{1},\ldots,c_{k-1},c_{k},\ldots)\\
(c_{1},c_{2},\ldots)
\end{aligned}
 & \begin{aligned} & \xrightarrow{\mu kc_{k}}(c_{1},\ldots,c_{k-1}+1,c_{k}-1,\ldots)\\
 & \xrightarrow{\mu c_{1}}(c_{1}-1,c_{2},\ldots)
\end{aligned}
 & \left.\vphantom{\begin{array}{c}
(c_{1})\\
(c_{1})
\end{array}}\right\} \, & \text{Death}
\end{align*}
If we assume detailed balance between pairs of states \deleted{subset
  of for which the maximum clone size is \protect{$K$
    ($c_{k>K}=0$)}}\added{with maximum clone size \protect{$K$}},
\deleted{the we recurse the equations}\added{we can recurse the relations}
\[
\mu c_{k}kP(c_{1},\ldots,c_{k-1},c_{k},\ldots)=
r(N)(k-1)\left(c_{k-1}+1\right) P(c_{1},\ldots,c_{k-1}+1,c_{k}-1,\ldots)
\]
for $2\leq k\leq K$ down to the states
\[
\mu c_{1} P(c_{1},\vec{0}) =\alpha P(c_{1}-1,\vec{0})
\]
to give
\begin{equation}
P(\vec{c})=\frac{1}{Z_{\alpha,\mu}}\frac{\alpha^{C}}{\mu^{N}}
\frac{\prod_{n=1}^{N-C}r(N-n)}{\prod_{i=1}^{\infty}i^{c_{i}}c_{i}!}.
\label{PC-bis}
\end{equation}
Using these chosen pairs of states to impose detailed balance, we find
a unique distribution $P(\vec{c})$.  However, \deleted{when} this
\added{form of} $P(\vec{c})$ \added{will not obey detailed balance
  between all pairs of states.}  \deleted{is plugged back into states
  that were not used to derive it, for example $P(c_{1}\geq 1,
  c_{2}\geq 1,\ldots)$, detailed balance is violated.}  For example,
balancing the transitions
\[
(c_{1},c_{2},\ldots)\xrightleftharpoons[\alpha]{\mu c_{1}}(c_{1}-1,c_{2},\ldots)
\]
would also require
\[
\mu c_{1}P(c_{1},c_{2}\geq 1,\ldots) = \alpha P(c_{1}-1,c_{2}\geq 1,\ldots).
\]
However, using the $P(\vec{c})$ from Equation~(\ref{PC-bis}), we find
\[
\frac{\mu c_{1}P(c_{1},c_{2}\geq 1,\ldots)}{\alpha P(c_{1}-1,c_{2}\geq 1,\ldots)}=
\frac{r(C-1)}{r(N-1)} \neq 1
\]
because generally, $N\neq C$. Remarkably, the analogous exercise for
the BDICC model where $\mu = \mu(N)$ does satisfy detailed balance
between all pairs of states and the $P(\vec{c})$ we derived for the
BDICC model, Equation~(\ref{PvecCC}), is exact.

\bibliographystyle{plainnat}
\bibliography{bibliography}

\end{document}